\documentclass[a4paper,12pt,twoside]{article}

\usepackage{amsfonts}
\usepackage{amsmath}
\usepackage{amsopn}
\usepackage{amsthm}
\usepackage{amssymb}
\usepackage{graphicx}
\usepackage{rotating}
\usepackage[]{natbib}
\usepackage[english]{babel}
\usepackage[utf8]{inputenc}
\usepackage[colorlinks=true,citecolor=blue,linkcolor=blue]{hyperref}
\usepackage{xspace}
\usepackage{setspace}
\usepackage{dsfont}
\usepackage{booktabs}
\usepackage[T1]{fontenc}
\usepackage{anysize}
\usepackage{enumerate}
\usepackage[labelfont=bf, labelsep=period]{caption}
\marginsize{1.87cm}{1.87cm}{1.87cm}{1.87cm} 
\usepackage{fancyhdr}
\usepackage{floatrow}
\newfloatcommand{capbtabbox}{table}[][\FBwidth]
\usepackage{algorithm}
\RequirePackage{txfonts}
\usepackage{times}
\usepackage{bm}
\usepackage{multirow}
\usepackage{wrapfig}

\newtheorem{Proposition}{Proposition}

\newtheorem{Remark}{Remark}

\DeclareMathOperator{\re}{\mathbb{R}}
\DeclareMathOperator{\simdist}{\stackrel{\mathcal{D}}{\sim}}
\newcommand{\ind}{\mathds{1}}
\newcommand{\E}{\mathbb{E}}
\newcommand{\abs}[1]{\left|#1\right|}
\DeclareMathOperator{\na}{\mathbb{N}}
\newcommand{\Prob}{\mathbb{P}}

\allowdisplaybreaks	

\begin{document}

\def\figureautorefname{Figure}
\def\sectionautorefname{Section}
\def\subsectionautorefname{Section}
\def\subsubsectionautorefname{Section}
\def\Propositionautorefname{Proposition}

\title{An automatic robust Bayesian approach to principal component regression}

\author{Philippe Gagnon $^{1}$, Myl\`{e}ne B\'{e}dard $^{2}$, Alain Desgagn\'{e} $^{3}$}

\maketitle

\thispagestyle{empty}

\noindent $^{1}$Department of Statistics, University of Oxford, United Kingdom.

\noindent $^{2}$Department of Mathematics and Statistics, Universit\'{e} de Montr\'{e}al, Canada.

\noindent $^{3}$Department of Mathematics, Universit\'{e} du Qu\'{e}bec \`{a} Montr\'{e}al, Canada.

\begin{abstract}
Principal component regression uses principal components as regressors. It is particularly useful in prediction settings with high-dimensional covariates.
The existing literature treating of Bayesian approaches is relatively sparse. We introduce a Bayesian approach that is robust to outliers in both the dependent variable and the covariates. Outliers can be thought of as observations that are not in line with the general trend. The proposed approach automatically penalises these observations so that their impact on the posterior gradually vanishes as they move further and further away from the general trend, corresponding to a concept in Bayesian statistics called \textit{whole robustness}. The predictions produced are thus consistent with the bulk of the data. The approach also exploits the geometry of principal components to efficiently identify those that are significant. Individual predictions obtained from the resulting models are consolidated according to model-averaging mechanisms to account for model uncertainty. The approach is evaluated on real data and compared to its nonrobust Bayesian counterpart, the traditional frequentist approach, and a commonly employed robust frequentist method. Detailed guidelines to automate the entire statistical procedure are provided. All required code is made available, see \href{https://arxiv.org/abs/1711.06341}{ArXiv:1711.06341}.
\end{abstract}


\section{Introduction}\label{sec_intro_pcr}

\begin{wrapfigure}{r}{0.36\textwidth}
\begin{center}
\vspace{-11mm}
\includegraphics[width=0.75\textwidth]{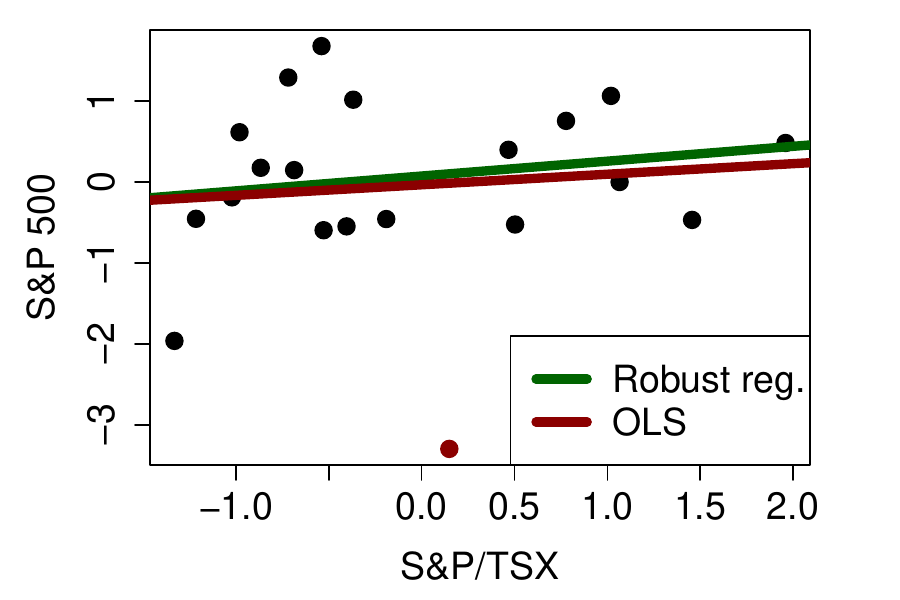}
\end{center}
\vspace{-8mm}
\caption{January 2011 daily returns}\label{fig_0} 
\vspace{-0mm}
\end{wrapfigure}

In statistical analyses, information carried by several variables is commonly summarised to allow visualisation or model estimation when the number of variables makes it unstable or impossible. For instance, S\&P 500 and S\&P/TSX respectively summarise the stock prices of 500 and about 250 large companies domiciled in the United States and Canada, and are commonly used to portray the American and Canadian economies. \autoref{fig_0} illustrates the relationship between the January 2011 daily returns of these two indices. The scatter plot is further summarised using two different linear regression models that respectively yield a robust regression line (in green) and an ordinary least squares regression line (in red). Given that different summaries (other than S\&P 500 or S\&P/TSX in our example) lead to different data points and therefore different regressions, one might however wonder whether the available or natural summaries are necessarily suitable for the tasks at hand.

Principal component regression (PCR) is the name given to a linear regression model using principal components (PCs) as regressors. It is based on a principal component analysis (PCA), which is commonly used to summarise the information contained in covariates. The principle is to find new axes in the covariate space by exploiting the correlation structure between the covariates, and then encode the covariate observations in that new coordinate system. The resulting variables, called principal components (PCs), are linearly independent and have the remarkable property that the first $q$ PCs retain the maximum amount of information carried by the original observations (compared to any other $q$-dimensional summary). Regrouping correlated variables to produce linearly independent ones is appealing in a linear regression context, as strongly correlated variables are known to carry redundant information, leading to unstable estimates. Companies within the same economic sector in stock market indices like S\&P 500 and S\&P/TSX are an example of such correlated variables. Linear independence also allows visualising the relationship between the dependent variable and the PCs by plotting the dependent variable against each of the PCs.

Due to the loss in the interpretability of the inference results engendered by transforming covariates, PCR is mainly used in a prediction context. It can nevertheless be useful for clarifying the underlying structure in the original covariates, as shown in \cite{West2003}. In this paper, we consider a Bayesian prediction framework and address four issues; those are described below.

\textbf{Robustness against outliers.} It is common knowledge that OLS (ordinary least squares) estimates become significantly contaminated in presence of outliers. In \autoref{fig_0}, the OLS regression line (in red) is pulled below the robust regression line (in green) by the outlier (red dot). OLS estimates make the assumption that errors are normally distributed, which affects both the linear regression and the PCA (see \cite{vsmidl2007bayesian}). In presence of outliers, the slimness of normal tails causes a shift in the posterior so as to incorporate the information carried by all the data. The posterior may thus find itself concentrated in an area that is not supported by any source of information, be it the outliers or the bulk of the data. This translates, for instance, into predictions that are not in line with either of these two groups.

The natural solution to this problem is to assume an error distribution with heavier tails, and therefore more adapted to the possible presence of outliers. The Student distribution becomes an obvious choice as it leads to a straightforward implementation of the Bayesian regression approach via the Gibbs sampler \citep{1984west431}. Using a heavy-tailed distribution like the Student however only allows attaining partial robustness \citep{andrade2011bayesian}, which may lead to regression coefficients with inflated variances, and ultimately contaminated model selection. Relying on an uncontaminated model selection procedure is crucial in our framework as the identification of important PCs relies on it.

%
%
It was recently proved in \cite{gagnon2018regression} that model selection in linear regression is uncontaminated when a super heavy-tailed error distribution is instead assumed. We follow this path, and based on that strategy of using super heavy-tailed distributions, introduce a new class of wholly robust Bayesian PCA. Hereafter, whole robustness refers to an approach that automatically penalises observations that are not in line with the general trend, so that their impact on the posterior distribution gradually vanishes as they move further from that trend.
 The assumed super heavy-tailed density matches the standard normal outside of the tails, which makes the approaches efficient. The resemblance between the two densities helps us to design the computational tools.

\textbf{Selection of significant PCs.} The selection of pertinent PCs to be included in our robust regression model is based on model selection and in line with the methods used in \cite{wang2012PCR} and \cite{tipton2017PCR}. Ours however differs in that we do not use the stochastic search variable selection (see \cite{George1993SSVS}), which is the common tool to discriminate among a large number of (typically correlated) regressors. We instead take advantage of the linear independence among PCs to quickly exclude the irrelevant ones, leading to the following two-step approach. We first evaluate the individual relevance of each PC through Bayes factors, after which the retained PCs are used to propose a sequence of nested models. The joint posterior of these models and their parameters is next computed. Observations for the dependent variable are predicted by accounting for model uncertainty through model averaging (see, for instance, \cite{raftery1997bayesian} and \cite{hoeting1999bayesian}).

\textbf{Automatic and efficient implementation.} 
Our approach to attain whole robustness (which consists in assuming super heavy-tailed error distributions)
however prevents us from having access to full conditional distributions and, therefore, to using Gibbs sampler.  For the robust PCA, we then propose a simplified computational scheme based on point estimates. The model posterior probabilities are however required in the linear regression stage of the statistical analysis, and so we turn to the reversible jump (RJ) algorithm to obtain estimates of these probabilities. The RJ sampler is a Markov chain Monte Carlo (MCMC) method introduced by \cite{green1995reversible} that allows to directly sample from the joint posterior of the models and their parameters. The efficiency of such samplers relies heavily on the design of the functions required for the implementation. We provide a detailed procedure to automatically implement an efficient RJ algorithm.

\textbf{Prior specification.} It is often difficult, in PCR, to specify meaningful priors on the models and their parameters. For this reason, noninformative priors are commonly favoured. The simplest noninformative structure is arguably the improper Jeffreys priors on the parameters of all models, along with
a uniform prior on the models. With such a prior structure, one might wonder whether the so-called Jeffreys-Lindley paradox \citep{lindley1957paradox, jeffreys1967prob}, representing inconsistent model selection results, may arise. We show that this is not the case and adopt that structure.

\textbf{Structure of the paper.} The general model is described in \autoref{sec_context_pcr}. Nonrobust normal PCA and regression approaches are presented in \autoref{sec_first_sit}, followed by their robust counterparts, representing the proposed methodology, in \autoref{sec_second_sit}. In particular, the proposed robust PCA is discussed in \autoref{sec_pca_robust}, while the robust linear regression is addressed in \autoref{sec_robust_reg}. \autoref{sec_RJ_PCR} presents the RJ sampler and then \autoref{sec_opt_impl_PCR} focuses on automating its implementation. The stock market indices example is revisited in \autoref{sec_real_data} where all the features of the proposed robust approach are illustrated. The validity of our prior structure is addressed in the supplementary material (\autoref{sec_supp}) as this part is 
not required to understand and implement the proposed methodology.

\section{Principal component regression}\label{sec_context_pcr}

Consider that we have access to a rank $r\in\{1,2,\ldots\}$ matrix $\mathbf{C}\in\re^{n\times p}$ containing $n\in\{1,2,\ldots\}$ observations from $p\in\{1,2,\ldots\}$ standardised covariates.
A PCA is then performed on this data set.
It will be seen that standardisation and PCA in the proposed robust approach are different from those in its nonrobust counterpart. We thus defer details about these steps to later sections.

Denote by $\mathbf{Z}_q$ the matrix of rank $q\leq r$ arising from either dimension reduction technique (nonrobust or robust PCA). The design matrix $\mathbf{X}:=(x_{ij})$ is constructed by simply grafting a column vector of 1's to the matrix $\mathbf{Z}_q$. For simplicity, we will refer to this extra column of $\mathbf{X}$ as the first component. The PCs are thus contained in the following columns, and $d:=q+1$ denotes the number of columns of $\mathbf{X}$.

We wish to study the relationship between a dependent variable with data points $Y_1,\ldots,Y_n\in\re$ and the PCs in order to predict values for the former. 
We start from the premise that the relationship is linear:
\begin{equation}\label{reg_mod}
        Y_i=\mathbf{x}_{i,K}^T \boldsymbol\beta_K+\epsilon_{i,K},\quad i=1,\ldots,n, \ \ K\in\{1,\ldots,\text{K}_{\text{max}}\},
\end{equation}
where $K$ is the model indicator, $\text{K}_{\text{max}}$ is a positive integer representing the number of models considered, and $\epsilon_{1,K},\ldots,\epsilon_{n,K}\in\re$ are the errors associated to Model $K$. 
The vector of observed PCs included in Model $K$ satisfies $\mathbf{x}_{i,K} := \{ x_{ij} : j \in I_K \}$, where $I_k\subseteq \{1,\ldots,d\}$ is a vector whose elements indicate which PCs are included in Model $K=k$. For instance, $I_1$ is associated to Model $1$ which, in this paper, always corresponds to the model containing only the intercept ($I_1:=\{1\}$). The $d_K$-dimensional vector of regression coefficients associated to Model $K$ is $\boldsymbol\beta_K:=(\beta_{1,K},\ldots,\beta_{d_K,K})^T\in\re^{d_K}$, where $d_K$ is the cardinality of $I_K$.  As is typically done in Bayesian linear regression, we assume that $\epsilon_{1,K},\ldots,\epsilon_{n,K}$ and $\boldsymbol\beta_K$ are $n+1$ conditionally independent random variables given $(K,\sigma_K)$, with $\sigma_K > 0$ being the scale parameter of the errors of Model $K$. The conditional density of $\epsilon_{i,K}$ is given by
   \begin{equation*}
    \epsilon_{i,K} \mid K,\sigma_K,\boldsymbol\beta_K\, \ \stackrel{\mathcal{D}}{=}\, \ \epsilon_{i,K}\mid K,\sigma_K\, \ \simdist\, \ (1/\sigma_K)f(\epsilon_{i,K}/\sigma_K)\,,\quad i=1,\ldots,n.
  \end{equation*}

  Even though we assume a 
  linear relationship between the dependent variable and regressors in \eqref{reg_mod}, we remain realistic and adopt George Box's point of view, which says that all models are wrong, but that some are useful. The degree of usefulness represented by the model fits will presumably be reflected in the posterior model probabilities.

To study the relationship between the dependent variable and the PCs, we first identify the statistically relevant PCs. The individual contribution of the various components is assessed using Bayes factors.  Specifically, we consider in the first step of the statistical analysis the $d$ models associated to $I_1:=\{1\}, I_2:=\{1, 2\}, \ldots, I_d = \{ 1, d \}$, and compare each of Models 2 through $d$ to Model 1. The PCs associated to Bayes factors greater than a given threshold are retained in the second step of the statistical analysis; the others are discarded.

In the second step of the analysis, we consider the sequence of nested models arising from the statistically significant PCs and find the posterior probabilities of these models, along with their parameter estimates. For instance, if the first, second and fourth PCs are the only ones deemed relevant, the sequence of models is  $I_1:=\{1\}$, $I_2:=\{1, 2\}$, and $I_3:=\{1, 2, 4\}$. Considering only a sequence of nested models is natural in our context, as PCA generates components that carry less and less information about the original covariates; that also simplifies subsequent computations.

Finding posterior probabilities and parameter estimates is achieved by sampling from the joint posterior distribution of $(K,\sigma_K,\boldsymbol\beta_K)$ given $\mathbf{y}:=(y_1,\ldots,y_n)^T$, denoted by $\pi(k,\sigma_k,\boldsymbol\beta_k \mid \mathbf{y})$, where the domain of $k$ depends on which step of the analysis is performed (and, for the second step, on the results of the previous step). Once estimates are obtained in the second step, values for the dependent variable can be predicted through model-averaging mechanisms.

\section{Normal nonrobust models}\label{sec_first_sit}

\subsection{Traditional principal component analysis}\label{sec_pca_nonrobust}

Several strategies allow retrieving the usual PCA from 
estimates of statistical models (see, e.g., \cite{Tipping1999probabilisticPCA} and \cite{vsmidl2007bayesian}). These methods assume that $\mathbf{C}$ has been generated from a linear model with normal errors. One can thus view PCs as point estimates and conduct a full Bayesian analysis of the model. We follow here the approach of \cite{vsmidl2007bayesian}; its presentation will facilitate the introduction of the robust PCA model as it will be analogously defined in \autoref{sec_pca_robust}.

The singular value decomposition allows expression of the matrix $\mathbf{C}$ as $\mathbf{Z}\mathbf{L}\mathbf{A}^T$, where $\mathbf{A}$ is a $p\times r$ matrix whose columns are the eigenvectors $\mathbf{v}_1,\ldots, \mathbf{v}_r$ of the sample correlation matrix of $\mathbf{C}$ with corresponding eigenvalues $\lambda_1\geq\lambda_2\geq\ldots\geq\lambda_r$, $\mathbf{L}$ is a $r\times r$ diagonal matrix with diagonal entries given by (up to a constant) $\lambda_1,\lambda_2,\ldots,\lambda_r$, and $\mathbf{Z}$ is a $n\times r$ matrix whose $j$-th column is given by $\lambda_j^{-1/2} \mathbf{C} \mathbf{v}_j$; see \cite{jolliffe2011principal} for instance. The PCs are traditionally defined as the vectors $\mathbf{C} \mathbf{v}_j$. We consider hereafter that the eigenvalues $\lambda_j$ are the sample variances of the PCs. The vectors $\lambda_j^{-1/2} \mathbf{C} \mathbf{v}_j$ therefore correspond to standardised PCs. Recall that the PCs are additionally pairwise orthogonal.

With $q < r$, let $\mathbf{Z}_q$ and $\mathbf{A}_q$ be the matrices comprised of the first $q$ columns of $\mathbf{Z}$ and $\mathbf{A}$, respectively, and $\mathbf{L}_q$ be the diagonal matrix with diagonal entries given by $\lambda_1,\ldots,\lambda_q$. If we want to further reduce the dimension of $\mathbf{Z}$ to $n\times q$, and therefore approximately reconstruct $\mathbf{C}$, \cite{vsmidl2007bayesian} present a model and a set of assumptions under which the maximum likelihood solution that arises is the anticipated matrix $\mathbf{Z}_q$. The model is
\begin{align}\label{eqn_model_PCA}
 \mathbf{C}=\mathbf{M}+\mathbf{E},
\end{align}
where $\mathbf{M}$ is assumed to have rank $q$ (and can therefore be decomposed using the singular value decomposition as above), and entries of $\mathbf{E}$ are assumed to be independently distributed as $\mathcal{N}(0,\eta^2)$, $\eta>0$. The maximum likelihood estimate (MLE) of $\mathbf{M}$ is $\mathbf{Z}_q\mathbf{L}_q\mathbf{A}_q^T$. This follows from the fact that $\mathbf{Z}_q\mathbf{L}_q\mathbf{A}_q^T$ minimises the total squared reconstruction error among rank $q$ matrices. The MLE corresponds to the maximum a posteriori (MAP) estimate when the prior is flat. We use the matrix $\mathbf{Z}_q$ to form our design matrix $\mathbf{X}$ in the nonrobust linear regressions.

It usually is good practice to cap the percentage of the total variation that is accounted for as above a certain threshold, eigenvectors are essentially numerical noise. In the numerical analyses we limit it to 95\%, meaning that $q$ is the maximum value such that $\sum_{j=1}^q \lambda_j/\sum_{j=1}^r \lambda_j\leq 0.95$.

\begin{Remark}
It is clear from \eqref{eqn_model_PCA} that $\mathbf{C}$ is viewed as a matrix containing observations from random variables. This may be confusing at first given that regressors are usually treated as known constants. In our case, the regressors are a function of $\mathbf{C}$; they are thus initially treated as observations from random variables in the PCA part of the statistical analysis. We next consider $\mathbf{Z}_q$ (or its robust version) as known constants in the regression part of the analysis. Our approach can thus be viewed as an approximation to the full and exact Bayesian analysis, in which all random unknown quantities, including $\mathbf{Z}_q, \boldsymbol\beta_K$, and $\sigma_K$, would be in linear models and estimated simultaneously, conditionally on $\mathbf{C}$ and $\mathbf{y}$. Our approach aims at simplifying the computation and interpretation of the statistical procedure.
\end{Remark}

\subsection{Ordinary least squares regression}\label{sec_OLS}

 Under the normality of the error distribution in the linear regressions (i.e.\ assuming that $f:=\mathcal{N}(0,1)$), the joint posterior $\pi(k,\sigma_k,\boldsymbol\beta_k \mid \mathbf{y})$ leads to closed-form expressions for the posterior model probabilities and parameter estimates. These expressions, detailed in \autoref{prop_posterior} below, are handy for comparing the results arising from our robust approach to those obtained under the normality assumption in the numerical analyses. They will also be used in the design of the RJ algorithm to sample from the posterior under the super heavy-tailed distribution assumption. Indeed, the super heavy-tailed distribution that we use is similar to the normal distribution, except in the tails. When there is no outlier, this thus leads to a posterior that is similar to that under normality. In the presence of outliers, the full posterior of the robust model is similar to the posterior based on the nonoutliers only (i.e.\ excluding the outliers) under normality. In either case, relying on the structure of the posterior under normality is therefore suitable for designing the RJ algorithm. 

 \begin{Proposition}\label{prop_posterior}
   Assume that $f:=\mathcal{N}(0,1)$ and let the conditional prior density of $(\sigma_K, \boldsymbol\beta_K)$ given $K$ be $\pi(\sigma_k, \boldsymbol\beta_k| k)\propto 1/\sigma_k$. Then, the posterior can be factorised as
   \begin{align*}
      \pi(k,\sigma_k,\boldsymbol\beta_k\mid \mathbf{y})&= \pi(k\mid \mathbf{y})\,\pi(\sigma_k\mid k,\mathbf{y})\,\prod_{j=1}^{d_k} \pi(\beta_{j,k}\mid k, \sigma_k,\mathbf{y}),
  \end{align*}
  where $k\in\{1,\ldots,\textnormal{K}_{\textnormal{max}}\},\sigma_k>0,\boldsymbol\beta_k\in\re^{d_k}$,
  \begin{align}\label{eqn_post_p}
    \pi(k\mid \mathbf{y})\propto\frac{\pi(k) \, \Gamma((n - d_k)/2) \, \pi^{d_k/2}}{\left(\|\mathbf{y} - \widehat{\mathbf{y}}_k \|_2^2 / (n - 1)\right)^{\frac{n-d_k}{2}}},
  \end{align}
  \begin{align*}
   \pi(\sigma_k\mid k,\mathbf{y})&=\frac{2^{1-\frac{n-d_k}{2}}\left(\|\mathbf{y} - \widehat{\mathbf{y}}_k \|_2^2\right)^{\frac{n-d_k}{2}}}{\Gamma((n-d_k)/2) \, \sigma_k^{n-d_k+1}} \,  \exp\left\{-\frac{1}{2\sigma_k^2}\, \|\mathbf{y} - \widehat{\mathbf{y}}_k \|_2^2\right\},
  \end{align*}
  $\beta_{1,K} \mid K, \sigma_K,\mathbf{y}\sim \mathcal{N}(\widehat{\beta}_{1,K}:=0, \sigma_K^2/n)$, and finally $\beta_{j,K} \mid K, \sigma_K,\mathbf{y}\sim\mathcal{N}(\widehat{\beta}_{j,K}:=\sum_{i=1}^n x_{iI_{j,K}} y_i/(n-1), \sigma_K^2/(n-1))$ for $j=2,\ldots,d_K$ (if $K\geq 2$). Here, $\|\cdot\|_2$ is the Euclidean norm, $\widehat{\mathbf{y}}_k:=\mathbf{x}_{i,k}^T \, \widehat{\boldsymbol\beta}_k$, $\widehat{\boldsymbol\beta}_k:=(\widehat{\beta}_{1,k},\ldots,\widehat{\beta}_{d_k,k})^T$, $I_{j,K}$ is the $j$-th component of $I_K$, and $\pi(k)$ is the prior of $K$. Note that the normalisation constant of $\pi(k \mid \mathbf{y})$ is the sum over $k$ of the expression on the right-hand side of \eqref{eqn_post_p}.
 \end{Proposition}

 \begin{proof}
 See the supplementary material (\autoref{sec_supp}).
\end{proof}

In our analyses, we use Bayesian model averaging to predict values for the dependent variable given sets of observations from the covariates. When normality is assumed, we can therefore use $\E[Y_{n+1} \mid\mathbf{y}]=\sum_k \pi(k \mid \mathbf{y}) \, \mathbf{x}_{n+1,k}^T \, \widehat{\boldsymbol\beta}_k$, where $\widehat{\boldsymbol\beta}_k$ is defined in Proposition \ref{prop_posterior}. Note that under normality, $\sigma_K^2 \mid K,\mathbf{y}$ has an inverse-gamma distribution with shape and rate parameters given by $(n-d_K)/2$ and $\|\mathbf{y} - \widehat{\mathbf{y}}_K \|_2^2/2$, respectively.

\section{Proposed robust models}\label{sec_second_sit}

The proposed solution to limit the impact of outliers in PCA and linear regression is simple: replace the traditional normality assumption on the error terms by a super heavy-tailed distribution assumption. The super heavy-tailed distribution that we use is the log-Pareto-tailed standard normal (LPTN) distribution with parameter $\rho\in (2\Phi(1) - 1, 1) \approx (0.6827, 1)$, where $\Phi$ is the cumulative distribution function of a standard normal. This distribution has been introduced in \cite{desgagne2015robustness} and is expressed as
 \begin{equation}\label{eqn_log_pareto_pcr}
  f(x):=\left\{
                                                    \begin{array}{lcc}
                                                      \varphi(x)  & \text{ if } & \abs{x}\leq \tau, \\
                                                      \varphi(\tau)\,\frac{\tau}{|x|}\left(\frac{\log \tau}{\log |x|}\right)^{\lambda+1} & \text{ if } & \abs{x}>\tau, \\
                                                    \end{array}
\right.
  \end{equation}
  where $x\in\re$. The terms $\tau>1$ and $\lambda>0$ are functions of $\rho$ and satisfy
  \begin{align*}
 & \tau:=\Phi^{-1}((1+\rho)/2) := \{\tau : \Prob(-\tau \leq Z \leq \tau)= \rho \,\text{ for }\, Z\, \simdist \, \mathcal{N}(0,1)\}, \\
 & \lambda:=2(1-\rho)^{-1}\varphi(\tau) \, \tau \log(\tau), \nonumber
 \end{align*}
 with $\varphi(\,\cdot\,)$ and $\Phi^{-1}(\,\cdot\,)$ respectively being the probability density function (PDF) and inverse cumulative distribution function of a standard normal. The parameter $\rho$ controls the size of the interval over which $f$ exactly matches the standard normal density (i.e.\ the interval $[-\tau, \tau]$). Outside of this area, the tails behave according to a log-Pareto density $(1/|x|)(\log|x|)^{-\lambda-1}$, hence its name.

Setting $\rho$ to $0.95$ has proved to be suitable for practical purposes, as addressed in \cite{desgagne2015robustness} for location-scale models and in \cite{gagnon2018regression} for linear regression. Accordingly, this is the value that will be used in our numerical analyses. Smaller values lead to improved robustness, but also to models that are further from normality (which then lead to discrepancies among estimations in the absence of outliers).

The theoretical result that motivates the use of super heavy-tailed distributions has been introduced in \cite{gagnon2018regression}. 
It establishes that, as outliers (because of extreme dependent and/or covariate observations) move further and further away from the general trend, the posterior distribution of $(K, \sigma_K, \boldsymbol\beta_K)$ arising from the whole data set converges towards the posterior of $(K, \sigma_K, \boldsymbol\beta_K)$ arising from the nonoutliers only. To prove this, it is however necessary to assume that there are at most $\lfloor n / 2 - (\max d_k - 1 / 2)\rfloor$ outliers in the data set, with $\lfloor\,\cdot\,\rfloor$ being the floor function. For a fixed $\max d_k$, this condition translates into a limiting breakdown point of $50\%$ as $n\longrightarrow\infty$.

As explained in \cite{gagnon2018regression}, these models have built-in robustness that resolves conflict in a sensitive way. It takes full consideration of nonoutliers and excludes observations that are undoubtedly outlying; in between these two extremes, it balances and bounds the impact of possible outliers. In other words, there is no need to explicitly identify outliers; the method automatically deals with the level of (un)certainty about the nature of the observations (nonoutliers, clear outliers or potential outliers), which is particularly valuable in high-dimensional and model selection problems.  
The robust models and their properties are the subject of a whole article. For brevity purposes, we refer the interested reader to \cite{gagnon2018regression} for more details.


\subsection{Robust principal component analysis}\label{sec_pca_robust}

Attempts at robustifying the traditional PCA model in \eqref{eqn_model_PCA} have been made by various authors (see, for instance, \cite{luttinen2009bayesian} and \cite{zhao2014robust}). They however follow the model specification of  \cite{Tipping1999probabilisticPCA} as opposed to that of \cite{vsmidl2007bayesian} (as we do here), and accordingly do not explicitly impose a rank constraint on the matrix $\mathbf{M}$ used to reconstruct $\mathbf{C}$. As mentioned in \autoref{sec_pca_nonrobust}, this constraint ensures that $\mathbf{M}$ can be decomposed as $\tilde{\mathbf{Z}}_q \tilde{\mathbf{L}}_q \tilde{\mathbf{A}}_q$, where $\tilde{\mathbf{Z}}_q$ and $\tilde{\mathbf{A}}_q$ have orthogonal columns (and are estimated by $\mathbf{Z}_q$ and $\mathbf{A}_q$ under the normal errors assumption). This orthogonality combined with the properties of PCA lead to the appealing geometric interpretation that those new axes are the best to reflect the information contained in $\mathbf{C}$. It also facilitates the statistical procedure for identifying relevant regressors. The price to pay for these advantages under the robust model is a significant increase in terms of computational complexity, as it becomes necessary to perform sampling and optimisation within the manifold of orthogonal matrices.  As an alternative to this computationally demanding route, we propose here an asymptotic approximation to a wholly robust PCA (as $n\longrightarrow\infty$ and outliers move further away from the general trend). An exhaustive analysis of the exact version (including its implementation) will be conducted separately.

In wholly robust PCA, the entries of the error matrix $\mathbf{E}:=(e_{ij})$ are such that $e_{ij} \mid\eta \, \ \simdist\, (1/\eta)g(e_{ij}/\eta)$, with $g$ the density of the LPTN. Under this error distribution assumption, we conjecture that a convergence result similar to that proved in \cite{gagnon2018regression} holds. In particular, the posterior distribution of $(\tilde{\mathbf{Z}}_q, \tilde{\mathbf{L}}_q, \tilde{\mathbf{A}}_q, \eta)$ (obtained from the covariate matrix $\mathbf{C}$ under LPTN errors)  converges towards the posterior of $(\tilde{\mathbf{Z}}_q, \tilde{\mathbf{L}}_q, \tilde{\mathbf{A}}_q, \eta)$ obtained from a new covariate matrix $\mathbf{C}^*$ and LPTN errors, as the outliers move away from the trend. Generally speaking, $\mathbf{C}^*$ is a matrix in which outlying covariate observations are vertically projected onto a regression plane that is obtained using the nonoutliers only.  
The proposed approximation to a wholly robust PCA makes use of the fact that the model with LPTN errors is similar to that with normal errors for the same reasons as \autoref{sec_OLS}, and thus essentially consists in computing the PCs using $\mathbf{C}^* \mathbf{v}_j^*$ as in \autoref{sec_pca_nonrobust}, with
$\mathbf{v}_j^*$ being the $j$-th eigenvector of a robust correlation matrix of $\mathbf{C}^*$. 
%
%
To better understand what happens, we consider an example containing a single PC which is simple enough for the wholly robust PCA model to be estimated. The orthogonality is indeed trivially verified given that there is only one column in $\tilde{\mathbf{Z}}_q$ and $\tilde{\mathbf{A}}_q$.

Suppose that $\mathbf{C}$ is a $21 \times 2$ matrix of observed covariates. Observations from the first covariate are $c_{i1} = i - 11$, $i=1, \ldots, 21$, and observed values from the second one are generated from the model $c_{i2} = c_{i1} + \epsilon_i$ with $\epsilon_i \sim \mathcal{N}(0,1)$, $i=1, \ldots, n$. \autoref{fig_robust_PCA} (a) illustrates the relationship between the observed covariates.

Let us now introduce an outlier in this sample by moving $(c_{21,1}, c_{21,2}) = (10, 10.92)$ to $(10,20)$; this sample is represented by the black dots in \autoref{fig_robust_PCA} (b). Applying a traditional PCA to these observed covariates and then using it to retrieve the matrix $\mathbf{C}$ yield the red dots in \autoref{fig_robust_PCA} (b); the reconstruction using the traditional PCA can be seen to rotate around the centre of the data as the outlier moves away from the trend. The wholly robust PCA approach leads to different results. The reconstruction of $\mathbf{C}$ using that approach is represented by the yellow dots in \autoref{fig_robust_PCA} (b).

Now, suppose that the outlier $(c_{21,1}, c_{21,2})$ is vertically projected onto a regression line that is obtained using the first 20 observed covariates (i.e.\ the nonoutlying points only); in other words, the outlier is replaced by its predicted value at $c_{21,1}$. Denote this new covariate matrix by $\mathbf{C}^*$. It is observed that as $j$ increases in $(c_{21,2}, c_{21,2}) = (10, 10.92+j)$, the posterior distribution of $(\tilde{\mathbf{Z}}_{q}, \tilde{\mathbf{L}}_q, \tilde{\mathbf{A}}_q,\eta)$ (obtained from $\mathbf{C}$, which includes the outlier) converges towards the posterior of $(\tilde{\mathbf{Z}}_{q}, \tilde{\mathbf{L}}_q, \tilde{\mathbf{A}}_q,\eta)$ obtained from $\mathbf{C}^*$. Applying the approximate robust PCA 
and then using it to reconstruct $\mathbf{C}$ yield the green dots in \autoref{fig_robust_PCA} (b).

In that figure, it is seen that the exact and approximate robust approaches (yellow and green dots) produce very similar results; the reconstruction of the outlier is however different under both approaches (we explain why it is the case and why it is not a problem in robust PCR in the following paragraphs). It turns out that as the outlier $(c_{21,2}, c_{21,2}) = (10, 10.92+j)$ reaches $(10,20)$, the posterior distribution of $(\tilde{\mathbf{Z}}_{q}, \tilde{\mathbf{L}}_q, \tilde{\mathbf{A}}_q,\eta)$ (based on $\mathbf{C}$) has essentially converged. Indeed, moving the outlier further upwards has no effect on the results from the robust approaches; this is obviously not the case for the traditional PCA, which pursues its rotation around the centre of the data.  For the data set in \autoref{fig_robust_PCA} (b), the squared reconstruction errors based on the nonoutliers only are 8.77 and 16.39 for the approximate robust and nonrobust PCA, respectively; the exact robust method yields a similar result to its approximate counterpart.

\begin{figure}[ht]
  \centering
  $\begin{array}{cc}
   \includegraphics[width=0.40\textwidth]{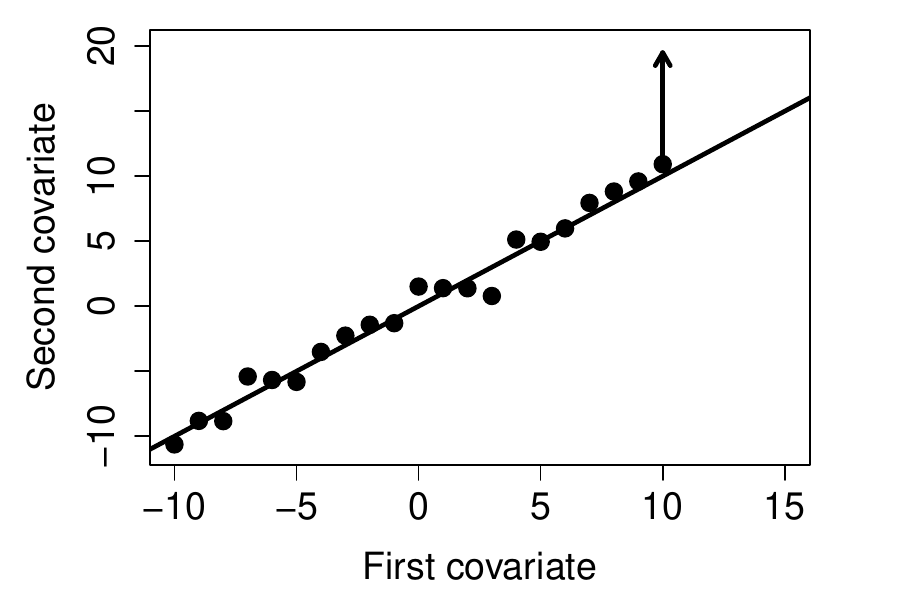} & \includegraphics[width=0.40\textwidth]{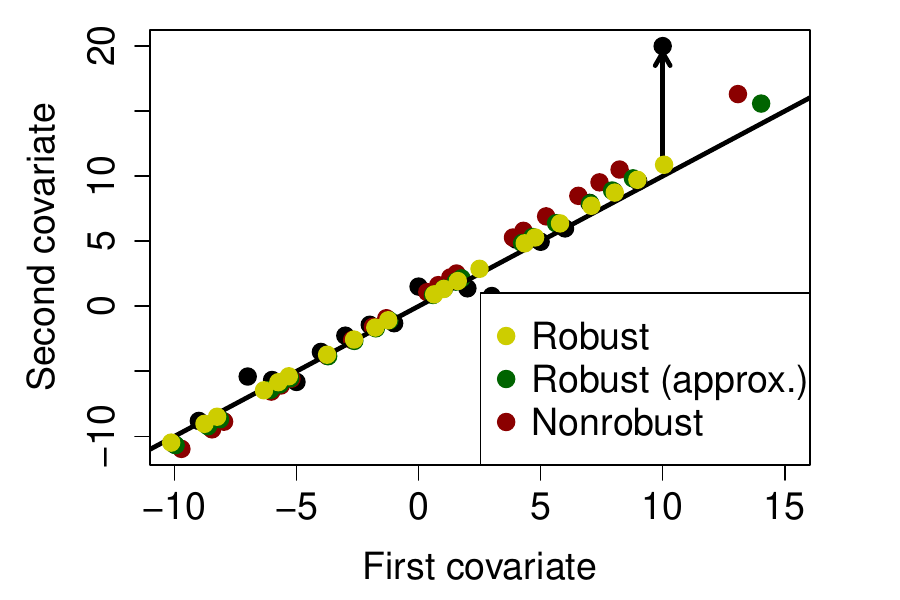} \cr
   \textbf{(a)} & \textbf{(b)}
  \end{array}$
  \vspace{-3mm}
\caption{(a) $n=21$ points generated from the model $c_{i2} = c_{i1} + \epsilon_i$ with $\epsilon_i\sim\mathcal{N}(0, 1)$; (b) data set with outlier, and reconstruction using one PC under the robust and nonrobust PCA; the lines $y=x$ are also depicted 
}\label{fig_robust_PCA}
 \end{figure}

We now detail the implementation of the approximate robust PCA.

\begin{enumerate}
\itemsep 0mm

 \item Standardise the columns of the original data set to obtain $\mathbf{C}$ using the robust location-scale model of \cite{desgagne2015robustness}, with an LPTN error distribution and $\rho:=0.95$. This model corresponds to the linear regression model with the intercept only and $f := \text{LPTN}$. Location and scale estimates $\widehat{\mu}_j$ and $\widehat{\sigma}_j$ are thus used to standardise Column $j$, $j=1,\ldots,p$.

 \item Compute robust correlations between all pairs of columns in $\mathbf{C}$ using the slope estimator of the robust simple regression model with an LPTN error distribution and $\rho:=0.95$. These correlations form the robust correlation matrix. For simplicity, we set the upper diagonal entries to $\widehat{\beta}_{j_1,j_2}$, which denote the robust correlations between the standardised Columns $j_1$ and $j_2$ where Column $j_2 > j_1$ plays the role of the dependent variable. We next make the matrix symmetrical and set its diagonal elements to 1.

 \item  Compute the PCs $\widehat{\mathbf{Z}}_{q}$ using $\mathbf{C} \widehat{\mathbf{v}}_j$, with $\widehat{\mathbf{v}}_j$ being the $j$-th eigenvector of the robust correlation matrix of $\mathbf{C}$.

\end{enumerate}

A mathematical justification of this approximation is presented in the supplementary material (\autoref{sec_supp}). It is shown that $\mathbf{C}$ is asymptotically equivalent to $\mathbf{C}^*$ except for the components where there are outliers. Also,  $\widehat{\mathbf{v}}_j$ is asymptotically equivalent to $\mathbf{v}_j^*$.
   Therefore there might be extreme values in the robust PCs, as there might be some in $\mathbf{C}$. If they exist, these extreme values will be handled by the robust linear regressions given their ability to deal with all types of outliers including leverage points.

Under the exact robust PCA approach, $\tilde{\mathbf{Z}}_{q}$ is directly estimated from the robust model; that represents the difference with the approximate method. The main advantage in using the approximate robust PCA is computational: the required estimates $\widehat{\mu}_j, \widehat{\sigma}_j$, and $\widehat{\beta}_{j_1,j_2}$ are easily obtained and can be computed in parallel. In our numerical experiments $\widehat{\mu}_j, \widehat{\sigma}_j$, and $\widehat{\beta}_{j_1,j_2}$ are maximum a posteriori (MAP) estimates with flat priors (corresponding to MLE). A second advantage is that the method allows automatic outlier detection. As in \cite{gagnon2018regression}, we compute estimates of the standardised residuals in the simple linear regressions as $z_{i}^{j_1,j_2}:=(c_{i,j_2}-\alpha_{j_1,j_2} - \beta_{j_1,j_2}c_{i,j_1})/\sigma_{j_1,j_2}$, using MAP estimates for instance, where $\alpha_{j_1,j_2}$ and $\sigma_{j_1,j_2}$ are the intercept and scale parameter in the robust model, respectively. One may then flag points with $|\widehat{z}_{i}^{j_1,j_2}| > 2.5$ (say), which is in line with classical recommandations (see \cite{GerviniYohai2002rewlse}). Note that the same principle applies for detecting outliers in the columns of $\mathbf{C}$ and, of course, in the multiple linear regressions used afterwards.

Finally note that the percentage of the total variation that is accounted for is capped at 95\%, as was the case with traditional PCA. The proposed method may lead to negative eigenvalues as robust correlation matrices are not correlation matrices per se. When this happens, we exclude the associated columns prior to setting $q$.

\subsection{Robust linear regressions} \label{sec_robust_reg}

The convergence result presented at the beginning of \autoref{sec_second_sit} ensures that posterior model probabilities and estimates of $(\sigma^K, \boldsymbol\beta^K)$ based on posterior quantiles (e.g.\ using posterior medians and Bayesian credible intervals) are robust to outliers. An analogous convergence result holds for the posterior expectations of the parameters, see \cite{gagnon2018regression}. Predictions for the dependent variable are then obtained by using 
$\sum_k \pi(k\mid \mathbf{y}) \, \mathbf{x}_{n+1,k}^T \, \widehat{\boldsymbol\beta}_k$ as in the nonrobust case, the difference being that probabilities and expectations are now computed with respect to the posterior arising from an LPTN error distribution. 
In Section~\ref{sec_RJ_PCR}, we describe the MCMC method used to approximate these probabilities and expectations; in Section \ref{sec_opt_impl_PCR}, we detail a procedure to efficiently implement this algorithm.

\subsubsection{Reversible Jump Algorithm}\label{sec_RJ_PCR}

As mentioned in \autoref{sec_pca_robust}, the price to pay for robustness is an increase in the complexity of the posterior. Parameters are however not restricted to a manifold in the linear regressions. Thus standard numerical approximation methods allow computing integrals with respect to posterior.
A commonly employed method for model selection and parameter estimation within the Bayesian paradigm is the RJ algorithm. This sampler allows simulation of the posterior distribution on spaces of varying dimensions, and can thus be used even if the number of parameters in the model is unknown.

The implementation of this sampler requires the specification of some functions, a step typically driven by the structure of the posterior. Recall that, whether there are outliers or not, the posterior under the super heavy-tailed LPTN distribution assumption has a structure similar to that expressed in \autoref{prop_posterior}. In other words, the regression coefficients should be nearly independent given $K$ and $\sigma^K$ and their values should not change dramatically from one model to another. In what follows, we borrow ideas from \cite{GAGNON201932}, in which an efficient RJ algorithm is built to sample from distributions with similar characteristics.

One iteration of the RJ sampler first randomly selects a model, and then proposes parameters for this model. This candidate model is then accepted as the next state of the Markov chain according to a specific probability; if it is rejected, the chain remains at the same state for another time interval. Specifically, given that the chain currently has $d_K$ components, the sampler that we use randomly selects one of three types of movements: update of the parameters; switch from Model $K$ to Model $K+1$ (with $d_{K+1} = d_K + 1$); switch from Model $K$ to Model $K-1$ (with $d_{K-1} = d_K - 1$).

The first step towards obtaining predictions is to identify the statistically relevant PCs. Recall that the individual contribution of each PC is evaluated by comparing the models $I_1=\{1 \}$ and $I_j =\{1, j\}$, $j=2, \ldots, d$. This first step of the statistical analysis requires $d-1=q$ runs of the RJ algorithm that can be performed in parallel. Performing model switches in those RJ samplers thus comes down to adding or withdrawing the $j$-th PC. Denote by $q^*$ the number of PCs associated to Bayes factors greater than the selected threshold; suppose that these statistically significant components are the $j_1$-th, $j_2$-th, \ldots, $j_{q^*}$-th PCs. The second step of the analysis then runs a single RJ sampler with $q^* +1$ nested models, ordered as follows : $I_1 = \{ 1 \}$ (intercept only), $I_2:=\{1,j_1\}$, \ldots, $I_{q^*+1}:=\{1,j_1,\ldots, j_{q^*}\}$. This ensures that the component added (removed) when switching models contains the most (least) information.

 The probability mass function used to randomly select the movement type at each iteration is
\begin{align}\label{def_g_pcr}
 g(j):=\begin{cases}
       \vartheta, \text{\quad if } j=1, \cr
       (1-\vartheta)/2, \text{\quad if } j=2,3,
      \end{cases}
\end{align}
where $0<\vartheta<1$ is a constant; the value of $\vartheta$ is discussed in Section \ref{sec_opt_impl_PCR}. At every iteration, an update of the parameters is thus attempted with probability $\vartheta$, while switches to Models $K+1$ and $K-1$ are attempted with probability  $(1-\vartheta)/2$ each.

Updating the parameters of Model $K$ is achieved here by using a $(d_K+1)$-dimensional proposal distribution centred around the current value of the parameter $(\sigma_K, \boldsymbol\beta_K)$ and scaled according to $\ell_K$, where $\ell_K$ is a positive constant given $K$. Each of the $d_K+1$ candidates is generated independently from the others, according to the one-dimensional strictly positive PDF $\varphi_i, i=1,\ldots,d_K+1$. Although the chosen PDF $\varphi_i$ usually is the normal density, we found the PDF in \eqref{eqn_log_pareto_pcr} to induce larger candidate steps and to result in a better exploration of the state space. We thus rely on this updating strategy in the analyses of \autoref{sec_real_data}. Note that one can easily simulate from \eqref{eqn_log_pareto_pcr} using the inverse transformation method.

A major issue with the design of RJ algorithms is that there might be a great difference between the ``good'' values of the parameters under Model $K$ and those under Model $K+1$ (or $K-1$). As explained in \autoref{sec_first_sit}, this is not a concern when there is no outlier, or when the same data points are diagnosed as outliers in Models $K$ and $K+1$; in these cases, the posterior under the LPTN is similar to that under normality. When observations are outliers with respect to Model $K$ but not Model $K+1$ (say), the posterior of Model $K$ is similar to that under normality excluding outliers, while the posterior of Model $K+1$ is similar to that under normality based on the whole sample. Therefore, when switching from Model $K$ to Model $K+1$, the parameters that were already in Model $K$ need to be moved to a position that is appropriate under Model $K+1$. Otherwise, this model switching will be less likely to be accepted, and the sampler will possibly require several iterations before the chain reaches high probability areas. Existing research has focused on that issue and found that it may result in inaccurate estimates, see \cite{brooks2003efficient}, \cite{al2004improving}, \cite{hastie2005towards}, and \cite{karagiannis2013annealed}.

Our strategy for resolving that issue is easily implemented. It consists in adding a vector $\mathbf{c}_{K+1}$ to the current parameters of Model $K$, so as to move these parameters to a suitable area under Model $K+1$. This leads to a candidate $(\sigma_{K+1},\boldsymbol\beta_{K+1}):=((\sigma_K,\boldsymbol\beta_{K})+\mathbf{c}_{K+1},u_{K+1})$ for Model $K+1$, where $(\sigma_K,\boldsymbol\beta_{K})$ is the current value of the parameter under Model $K$ and $u_{K+1}$ is a candidate for the added parameter $\beta_{d_{K+1},K+1}$, generated from an appropriate strictly positive PDF $q_{K+1}$. To avoid obtaining negative values for $\sigma_K$, we always set the first component of the vectors $\mathbf{c}_{i}$ to 0.

We now provide a pseudo-code to sample from $\pi(k,\sigma_k, \boldsymbol\beta_k\mid \mathbf{y})$ using the RJ sampler. In the next section, we specify the various inputs required to implement this algorithm.

 \begin{enumerate}
 \itemsep 0mm
  \item Initialise the sampler by setting $(K,\sigma_{K},\boldsymbol\beta_{K})(0)$. \newline \textit{Remark}:
  The number in parentheses beside a vector denotes the iteration.
  \item[\textbf{Iteration $m+1$.}\hspace{-25mm}]
  \item Generate $u\sim\mathcal{U}(0,1)$.
  \item[(a)\hspace{-4.5mm}] \hspace{4.5mm} If $u\leq \vartheta$, attempt an update of the parameters. Generate a candidate $\mathbf{w}_{K(m)}:=(w_1, \ldots, w_{d_{K(m)}+1})$, where $w_1\sim \varphi_1(\,\cdot \mid K(m),\sigma_{K}(m),\ell_{K(m)})$ and $w_i\sim \varphi_i(\,\cdot \mid K(m),\beta_{i-1,K}(m),\ell_{K(m)})$ for $i=2,\ldots,d_{K(m)}+1$. Generate $u_a  \sim\mathcal{U}(0,1)$; if
  $$
   u_a\leq \left(1\wedge\frac{(1/w_1)f(\mathbf{y}\mid K(m),\mathbf{w}_{K(m)})}{(1/\sigma_{K}(m))f(\mathbf{y}\mid (K,\sigma_{K},\boldsymbol\beta_{K})(m))}\right),
  $$
  where
  $$
   f(\mathbf{y}\mid k,\sigma_{k},\boldsymbol\beta_{k}):=\prod_{i=1}^n \frac{1}{\sigma_k}f\left(\frac{y_i-\mathbf{x}_{i,k}^T \boldsymbol\beta_k}{\sigma_{k}}\right),
  $$
  set $(K,\sigma_{K},\boldsymbol\beta_{K})(m+1)=(K(m),\mathbf{w}_{K(m)})$.
  \item[(b)\hspace{-4.5mm}] \hspace{4.5mm} If $\vartheta<u\leq \vartheta+(1-\vartheta)/2$, attempt adding a parameter to switch from Model $K(m)$ to Model $K(m)+1$. Generate  $u_{K(m)+1}\sim q_{K(m)+1}$ and $u_a \sim\mathcal{U}(0,1)$; if
    \begin{align*}
  u_a&\leq \left(1\wedge\frac{\pi(K(m)+1)f(\mathbf{y}\mid K(m)+1,(\sigma_{K},\boldsymbol\beta_{K})(m)+\mathbf{c}_{K(m)+1},u_{K(m)+1})}{\pi(K(m))f(\mathbf{y}\mid (K,\sigma_{K},\boldsymbol\beta_{K})(m))q_{K(m)+1}(u_{K(m)+1})}\right),
 \end{align*}
  set $(K,\sigma_{K},\boldsymbol\beta_{K})(m+1)=(K(m)+1,(\sigma_{K},\boldsymbol\beta_{K})(m)+\mathbf{c}_{K(m)+1},u_{K(m)+1})$.
  \item[(c)\hspace{-4.5mm}] \hspace{4.5mm} If $u> \vartheta+(1-\vartheta)/2$, attempt withdrawing the last parameter to switch from Model $K(m)$ to Model $K(m)-1$. Generate $u_a\sim\mathcal{U}(0,1)$; if
    \begin{align*}
  \hspace{-8mm} u_a&\leq\left(1\wedge\frac{\pi(K(m)-1)f(\mathbf{y}\mid K(m)-1,(\sigma_{K},\boldsymbol\beta_{K-})(m)-\mathbf{c}_{K(m)})q_{K(m)}(\beta_{d_K,K}(m))}{\pi(K(m))f(\mathbf{y}\mid (K,\sigma_{K},\boldsymbol\beta_{K})(m))} \right),
 \end{align*}
 where $(\sigma_{K},\boldsymbol\beta_{K-})(m):=(\sigma_{K},\beta_{1,K},\ldots, \beta_{d_K-1,K})(m)$, then
 set $(K,\sigma_{K},\boldsymbol\beta_{K})(m+1)=(K(m)-1,(\sigma_{K},\boldsymbol\beta_{K-})(m)-\mathbf{c}_{K(m)})$. 
 \item In case of rejection, set $(K,\sigma_{K},\boldsymbol\beta_{K})(m+1)=(K,\sigma_{K},\boldsymbol\beta_{K})(m)$.
 \item Go to Step 2.
 \end{enumerate}

It is easily verified that the resulting stochastic process $\{(K,\sigma_{K},\boldsymbol\beta_{K})(m):m\in\na\}$ is a $\pi(k,\sigma_k,\boldsymbol\beta_k\mid \mathbf{y})$-irreduci\-ble and aperiodic Markov chain. Furthermore, it satisfies the reversibility condition with respect to the posterior, as stated in the following proposition. Therefore, it is an ergodic Markov chain, which guarantees that the Law of Large Numbers holds.

\begin{Proposition}\label{prop_reversibility}
The Markov chain $\{(K,\sigma_{K}, \boldsymbol\beta_{K})(m):m\in\na\}$ arising from the RJ described above satisfies the reversibility condition with respect to the posterior $\pi(k,\sigma_k, \boldsymbol\beta_k\mid \mathbf{y})$.
\end{Proposition}

\begin{proof}
 See the supplementary material (\autoref{sec_supp}).
\end{proof}

 \subsubsection{Efficient implementation}\label{sec_opt_impl_PCR}

An optimal implementation of the RJ algorithm described above requires carefully selecting the various inputs: the PDFs $q_i$, the constants $\vartheta$ and $\ell_i$, and the vectors $\mathbf{c}_i$. Hereafter, ``optimal implementation'' or ``optimal design'' means that the generated Markov chain mixes as rapidly as possible, thus engendering least variable estimators.

In \cite{GAGNON201932}, a posterior structure similar to that expressed in \autoref{prop_posterior} is considered, and theoretical results leading to an optimal RJ algorithm are obtained. In that paper, the parameters of any given model are conditionally independent and identically distributed. An implicit assumption on the posterior studied is that distributions of parameters remain the same when switching from Model $K$ to Model $K+1$ (or $K-1$). The authors find asymptotically optimal values for $\vartheta$ and $\ell_K$ (as the number of parameters approaches infinity). They conjecture that their results are valid (to some extent) when the parameters are conditionally independent, but not identically distributed (for any given model). They also provide guidelines to suitably design the PDFs $q_K$. 
We use these results as a starting point in the design of our RJ algorithm.

In the settings of \cite{GAGNON201932}, the asymptotically optimal value for $\vartheta$ depends on the PDFs $q_i$. It is also empirically observed that for moderate values of $\text{K}_{\text{max}}$, selecting any value between 0.2 and 0.6 is almost optimal. We use $\vartheta:=0.6$ in the numerical analyses of \autoref{sec_real_data}, as $\text{K}_{\text{max}}$ is rather small (there are few models to visit). Generally speaking, larger values of $\vartheta$ leave the chain more time for exploring the parameters' state space between model switches. Based on several runs of the RJ algorithm, $\vartheta:=0.6$ is in fact nearly optimal for the data in \autoref{sec_real_data}.

If the parameters $(\sigma_K,\boldsymbol\beta_K)$ were independent and identically distributed for each model, the asymptotically optimal value for $\ell_K$ would be $\ell/\sqrt{d_K+1}$, with $\ell$ tuned to accept 23.4\% of candidates $\mathbf{w}_{K}$. When these assumptions are violated, the asymptotically optimal value for $\ell$ usually corresponds to an acceptance rate smaller than $0.234$ (see \cite{bedard2007weak} and \cite{bedard2017hierarchical}).
Considering this, and adding the fact that $d_{k}$ may be rather small, we recommend to perform trial runs to identify optimal values for all $\ell_k$. We use the $0.234$ rule within each model to initiate the process. In our analyses in \autoref{sec_real_data}, the optimal values for all $\ell_k$ correspond to an acceptance rate relatively close to $0.234$. 

We propose to specify the PDFs $q_i$ and vectors $\mathbf{c}_i$ through trial runs as well. Specifying these functions and vectors requires information about locations and scalings of regression coefficients for all models. We gather this information by running a random walk Metropolis algorithm for each model; this sampler may be seen as a RJ algorithm in which $\vartheta:=1$ (i.e.\ a sampler in which only updates of the parameters are proposed). The recommended procedure is now detailed.

\begin{enumerate}
\itemsep 0mm
 \item[\textbf{For each} $k\in\{1,\ldots,\text{K}_{\text{max}}\}$\textbf{:}\hspace{-44mm}]
 \item Tune $\ell_k$ such that the acceptance rate of candidates $\mathbf{w}_{k}$ is approximately $0.234$; denote this value by $\ell_k^{\text{start}}$.
 \item Select a sequence of values around $\ell_k^{\text{start}}$: $(\ell_{1,k},\ldots,\ell_{j_0,k}:=\ell_k^{\text{start}}, \ldots, \ell_{L,k})$, where $L$ is a positive integer.
 \item For each $\ell_{j, k}$, run a random walk Metropolis sampler initialised as follows: $(\sigma_{k}(0))^2\sim \text{Inv-}\Gamma$ with shape and rate given by $(n-d_k)/2$ and $\|\mathbf{y} - \widehat{\mathbf{y}}_k\|_2^2/2$, respectively; $\beta_{1,k}(0)\sim\mathcal{N}(\widehat{\beta}_{1,k}, (\sigma_{k}(0))^2/n)$, and $\beta_{j,k}(0)\sim\mathcal{N}(\widehat{\beta}_{j,k},(\sigma_{k}(0))^2/(n-1))$, $j=2,\ldots,d_k$ (if $k\geq 2$). Here, $\widehat{\mathbf{y}}_k$ is computed using a preliminary robust estimate  $\widehat{\boldsymbol\beta}:=(\widehat{\beta}_{1,k},\ldots, \widehat{\beta}_{d_k,k})$ (MAP estimate under the robust LPTN model for instance).
 \item For each $\ell_{j, k}$ ($j=1, \ldots, L$), estimate the location and scaling of each $\beta_{i,k}$ using the runs in Step (3). In particular, compute the mean (denoted by $m_{i,j}^k$) and standard deviation (denoted by $s_{i,j}^k$) of $\{\beta_{i,k}(m):m\in\{B+1,\ldots,T\}\}$, for $i=1,\ldots,d_k$, where $B$ is the length of the burn-in period and $T$ the number of iterations. Repeat for $\sigma_k$, denoting the means and standard deviations by $m_{\sigma,j}^k$ and $s_{\sigma,j}^k$. Measure the efficiency of the sampler with respect to $\ell_{j,k}$ using the sum of the integrated autocorrelation times (IAT) of $\{\sigma_{k}(m):m\in\{B+1,\ldots,T\}\}$ and $\{\beta_{i, k}(m):m\in\{B+1,\ldots,T\}\}$ for $i=1,\ldots,d_k$. Record the value $\ell_k^{\text{opt}}$ corresponding to the smallest IAT.
 \item If $\ell_k^{\text{opt}}$ corresponds to the lower or upper bound of the range defined in Step (2), i.e.\ $\ell_{1,k}$ or $\ell_{L,k}$, change the sequence of values for $\ell_k$ and repeat.
 \item For $i=1,\ldots,d_k$, compute the average of $\{m_{i,1}^k,\ldots,m_{i,L}^k\}$ (denoted by $m_{i,k}$) and $\{s_{i,1}^k,\ldots, s_{i,L}^k\}$ (denoted by $s_{i,k}$). Also compute the average of $\{m_{\sigma,1}^k,\ldots,m_{\sigma,L}^k\}$ (denoted by $m_{\sigma,k}$) and $\{s_{\sigma,1}^k,\ldots,s_{\sigma,L}^k\}$ (denoted by $s_{\sigma,k}$).
\end{enumerate}

These runs can be performed in parallel for computational efficiency, in an automatic procedure that allows users to retrieve the desired output at the end. Using this output, set $q_j$ ($j=2,\ldots,\text{K}_{\text{max}}$) equal to the distribution in (\ref{eqn_log_pareto_pcr}), with location and scale parameters given by $m_{d_j,j}$ and $s_{d_j,j}$, respectively. Also set $\mathbf{c}_{j}:=(0,m_{1,j}-m_{1,j-1},\ldots,m_{d_j-1,j}-m_{d_j-1,j-1})^T$, $j=2,\ldots,\text{K}_{\text{max}}$, and $\ell_k$ equal to $\ell_k^{\text{opt}}$ for all $k$.

The only inputs left to choose before implementing the RJ algorithm are the initial values for the model indicator and parameters. We recommend to generate $K(0)\sim \mathcal{U}\{1,\ldots,\text{K}_{\text{max}}\}$, $\sigma_{K}(0)$ from a normal truncated at 0 with mean $m_{\sigma,K(0)}$ and standard deviation $s_{\sigma,K(0)}$, and $\beta_{j,K}(0)\sim\mathcal{N}(m_{j,K(0)},s_{j,K(0)}^2)$, $j=1,\ldots,d_{K(0)}$. The analyses of \autoref{sec_real_data} rely on sequences of length $L=11$ for $\ell_k$ ($\ell_k^{\text{start}}$ is the median), $T=\,$100,000 iterations, and a burn-in period of length $B=\,$10,000 for the trial runs. When running the RJ sampler, we use 1,000,000 iterations and a burn-in period of length 100,000.

\section{Case study: prediction of returns for the S\&P 500}\label{sec_real_data}

In this section, we illustrate the performance of our robust approach on a real data set containing outliers. We provide a detailed analysis of the results and contrast them with those from other approaches to identify in which situations it is expected to perform better. The data set and context are described in \autoref{sec_data_des}, the competitors are presented in \autoref{sec_competitors}, while the section finishes with the result analysis and comparison in \autoref{sec_results}.

The use of super heavy-tailed distributions in linear regression has recently been introduced in \cite{DesGag2019}, where the special case of simple linear regressions through the origin was studied. The usual linear regression model was later analysed in \cite{gagnon2018regression}. Although theoretical results about model selection are presented in \cite{gagnon2018regression}, it is the first time that an illustration of the practical benefits is presented in that context.

\subsection{Data set and context description}\label{sec_data_des}

In this example, we model the January 2011 daily returns of the S\&P 500 by exploiting their potential linear relationship with some financial assets and indicators. We next use the estimated models to predict the February 2011 daily returns of this stock index; to this end, covariate observations on day $i$ will be used to predict the return of the S\&P 500 on day $i+1$. A detailed list of the 18 covariates considered is provided in the supplementary material (\autoref{sec_supp}); $n=19$ observations are available for model estimation. The full linear regression model with all covariates would have 20 parameters ($p=18$ regression coefficients for the covariates, to which we add the intercept and scale parameter). We perform robust and nonrobust PCA procedures, which are expected to be beneficial given that financial assets and indicators are likely to carry redundant information.

\subsection{Competitors}\label{sec_competitors}

The results are compared with those obtained under the normality of errors assumption (nonrobust Bayesian approach) to evaluate outlier protection performance. The classical frequentist approach and the robust frequentist approach of \cite{hubert2003robust} are also included in the comparison. The implementation of a robust frequentist approach allows contrasting the effects of our Bayesian robust PCA decomposition. The proposed model-based PC selection approach is also evaluated. 

In principle, to construct a PCR approach, one only needs a PCA and a linear regression method. There are of course many combinations of PCA and linear regression approaches possible. To keep the analysis and comparison simple, we restrict our attention to a single combination of these methods 
for each of the four classes of PCR approaches considered (robust/nonrobust and Bayesian/frequentist). The four combinations selected are arguably the best approaches in each of the four classes. In the robust frequentist approach, we use MM-regression \citep{Yohai1987MM}  as it offers one of the best available asymptotic \textit{breakdown point versus efficiency} tradeoffs.  
We nevertheless acknowledge that there exist several other good combinations; one could, for instance, use least trimmed squares estimators (LTS, \cite{Rousseeuw1985LTS}), M-estimators \citep{huber1973robust}, S-estimators \citep{rousseeuw1984robust}, or other more recent robust regression approaches like those of  \cite{agostinelli2013weighted} and \cite{atkinson2017robust}. We refer the reader to the recent robust regression comparisons presented in \cite{gagnon2018regression} and \cite{yuyao2017review}, which respectively focus on Bayesian and frequentist methods.

Model estimation is performed using each of the four mentioned approaches. In the robust and nonrobust Bayesian approaches, statistically significant PCs are identified by relying on Bayes factors with a threshold of 1. This means that when on average (over the parameter space) an individual PC improves the fit over the model consisting solely of the intercept, then it is included in the second stage of the statistical analysis (for building the nested models).

\subsection{Results and analysis}\label{sec_results}

 The average absolute deviations (AAD) between the predicted and actual February 2011 returns are reported in the first column of \autoref{table:realdata}. In the current financial context, it may also be of interest to predict whether the asset (S\&P 500 in our case) will go up or down the following day. Using the sign of our predicted returns, we find to be correct 10, 13, 10, and 11 times out of 19 under the normal, LPTN, classical and robust frequentist approaches; the success rates are reported in the second column of \autoref{table:realdata}.

\begin{table}[ht]
\begin{center}
\begin{tabular}{l rr}
\toprule
\textbf{Approach} &  \textbf{AAD} & \textbf{Sign prediction rate} \cr
\midrule
 Normal errors (nonrobust Bayesian) & 0.60 & 0.53 \cr
 LPTN errors (robust Bayesian) & 0.49 & 0.68 \cr
 Classical frequentist & 0.63 & 0.53 \cr
 Robust frequentist & 0.57 & 0.58 \cr
 \bottomrule
\end{tabular}
\end{center}
\vspace{-5mm}
 \caption{Prediction results for the February 2011 daily returns of the S\&P 500, using robust and nonrobust versions of Bayesian and frequentist approaches}\label{table:realdata}
\end{table}

We know that differences in the results obtained from the normal and LPTN models are essentially due to the presence of outliers. The outlier detection method described in \autoref{sec_pca_robust} indeed flags several observations, both in the PCA and linear regression steps. Each graph in \autoref{fig_2} illustrates linear relationships between a different pair of covariates;  the two linear relationships in a given graph are established using the normal and LPTN error distributions. We see that in the presence of outliers, the choice of error distribution obviously has a large impact on the trends obtained from the data. \autoref{fig_4} also depicts linear relationships using the normal and LPTN distributions, but this time between the dependent variable and some PCs. Specifically, 
the graphs on the top line picture linear relationships between the dependent variable and the second PC; in the left graph, the second PC was constructed using a robust PCA while in the right graph, that same PC was constructed using a traditional PCA. The exercise is then repeated with the fourth PC and produces the two graphs on the bottom line.

\begin{figure}[ht]
  \centering
  $\begin{array}{ccc}
   \hspace{-4mm} \includegraphics[width=0.36\textwidth]{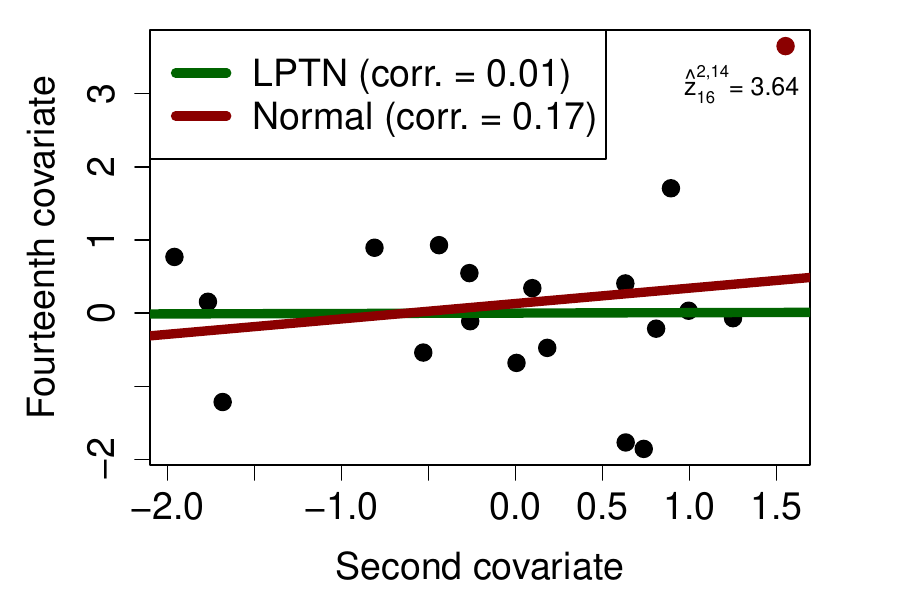} & \hspace{-7mm} \includegraphics[width=0.36\textwidth]{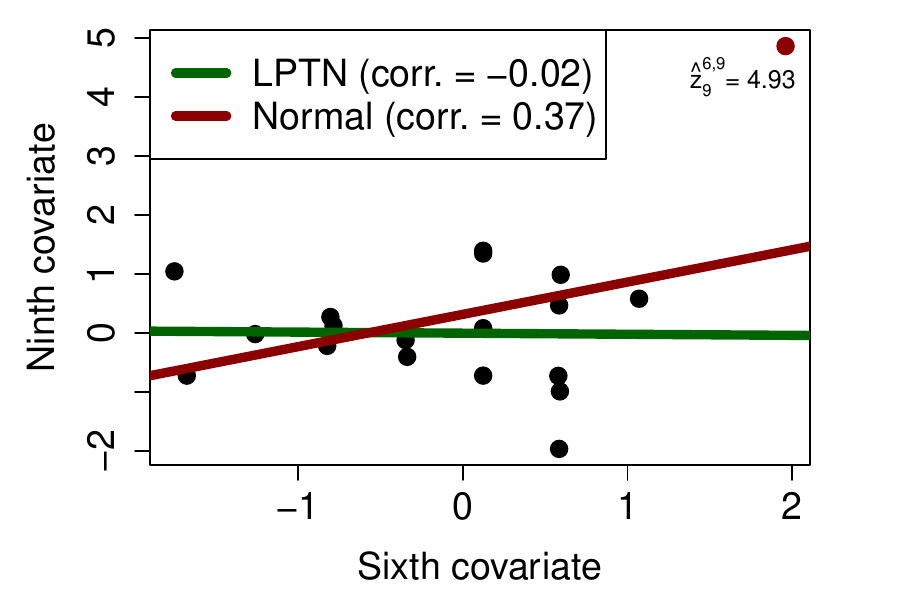} & \hspace{-7mm} \includegraphics[width=0.36\textwidth]{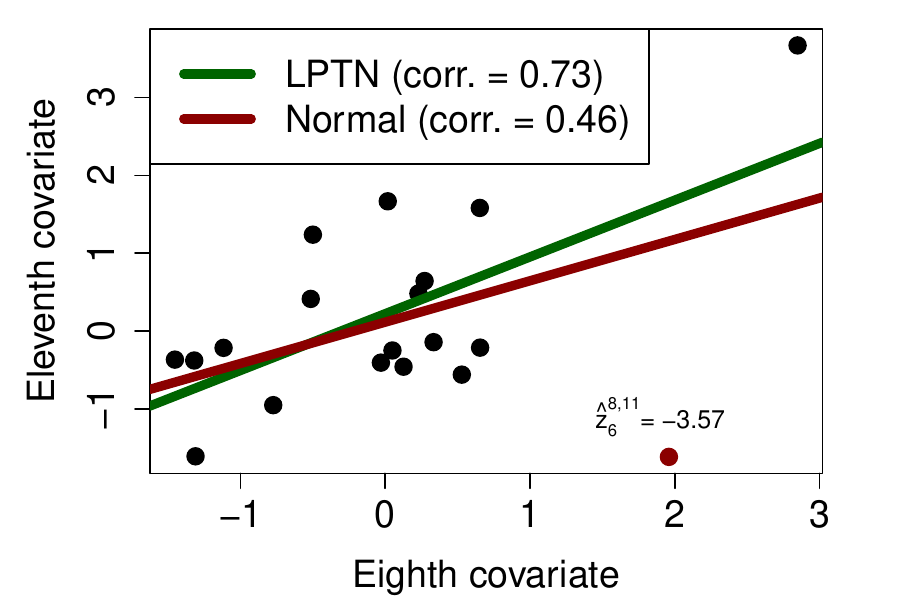} \cr
   \textbf{(a)} & \hspace{-5mm} \textbf{(b)} & \hspace{-5mm} \textbf{(c)}
  \end{array}$
  \vspace{-3mm}
\caption{Linear relationships between the (a) second and fourteenth covariates; (b) sixth and ninth covariates; (c) eighth and eleventh covariates}\label{fig_2}
 \end{figure}

From these figures, it is clear that contaminated correlation estimation in traditional PCA leads to an inferior assessment of the relationships between covariates (\autoref{fig_2}), which in turn leads to a different way of constructing the PCs (\autoref{fig_3}). By \textit{inferior}, we mean here that the trend does not reflect the behaviour of the majority of the observations, but rather consists in a poor compromise between that behaviour and the behaviour of outliers. The left graph of \autoref{fig_3} plots the differences arising from applying a robust PCA rather than a traditional one when computing the covariates weights used to construct the second PC. 
The right graph repeats the exercise for the fourth PC, which leads to even greater differences than for the second PC.

\begin{figure}[ht]
  \centering
  $\begin{array}{cc}
    \includegraphics[width=0.36\textwidth]{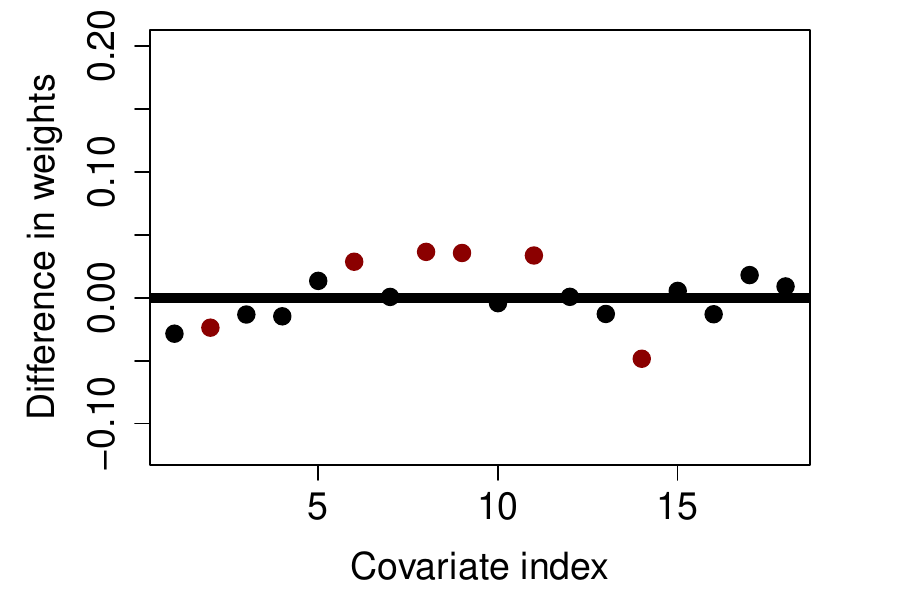} & \includegraphics[width=0.36\textwidth]{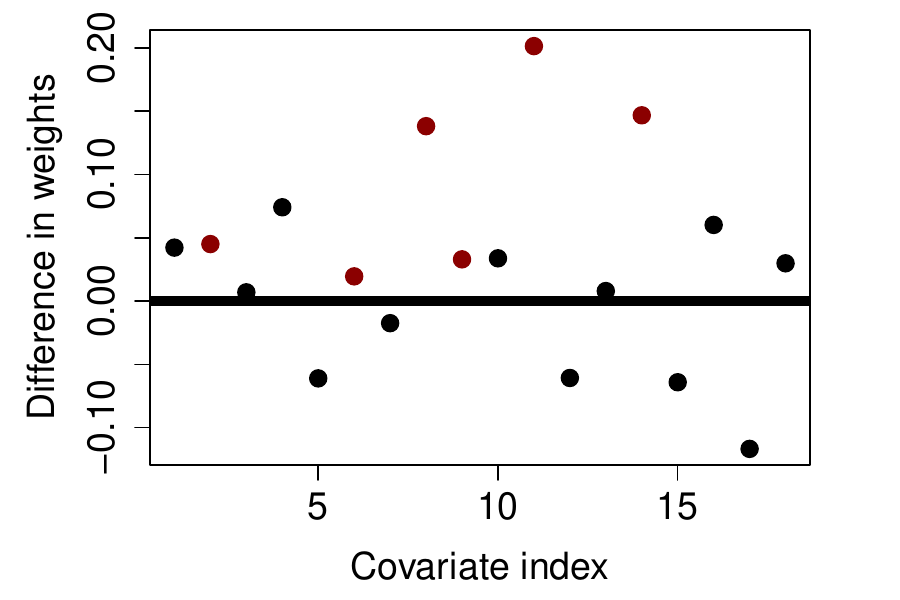} \cr
    \textbf{(a) Second PC} & \textbf{(b) Fourth PC}
  \end{array}$
  \vspace{-3mm}
\caption{Differences in covariate weights used to construct the (a) second PC and (b) fourth PC, when comparing the robust and traditional PCA; the six covariates considered in \autoref{fig_2} are shown in red (2nd, 6th, 8th, 9th, 11th, and 14th covariates)}\label{fig_3}
 \end{figure}

 \begin{figure}[ht]
  \centering
  $\begin{array}{cc}
    \includegraphics[width=0.36\textwidth]{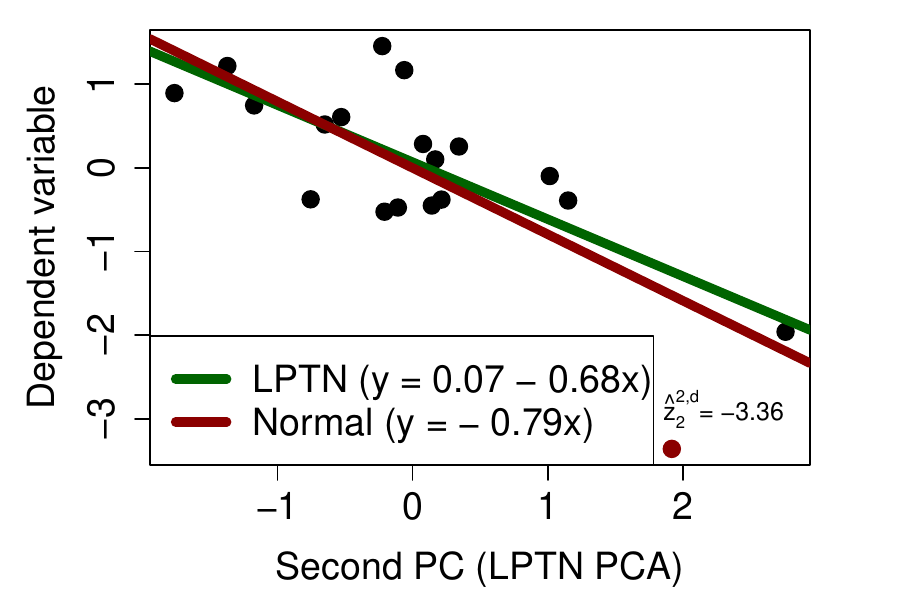} & \includegraphics[width=0.36\textwidth]{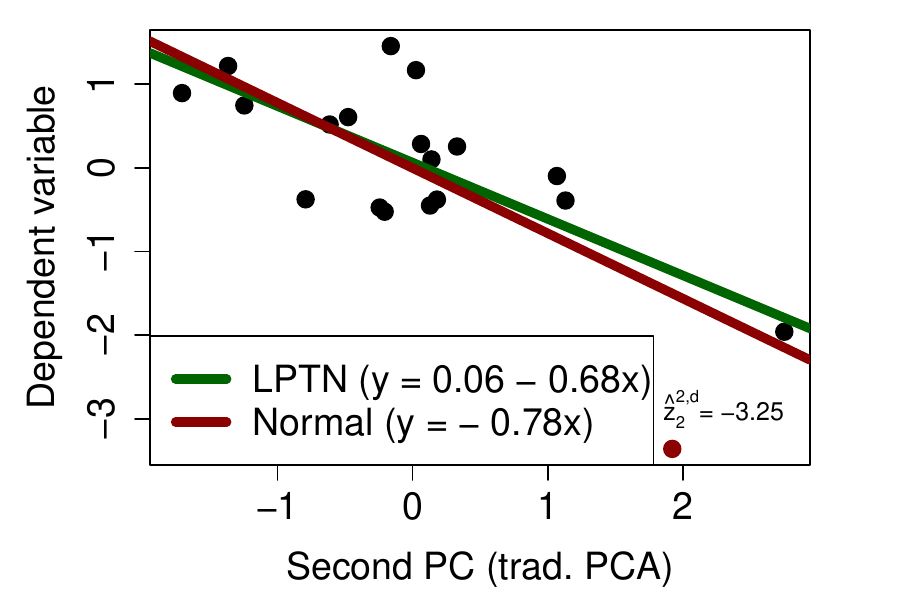} \cr
    \textbf{(a)} & \textbf{(b)} \cr
    \includegraphics[width=0.36\textwidth]{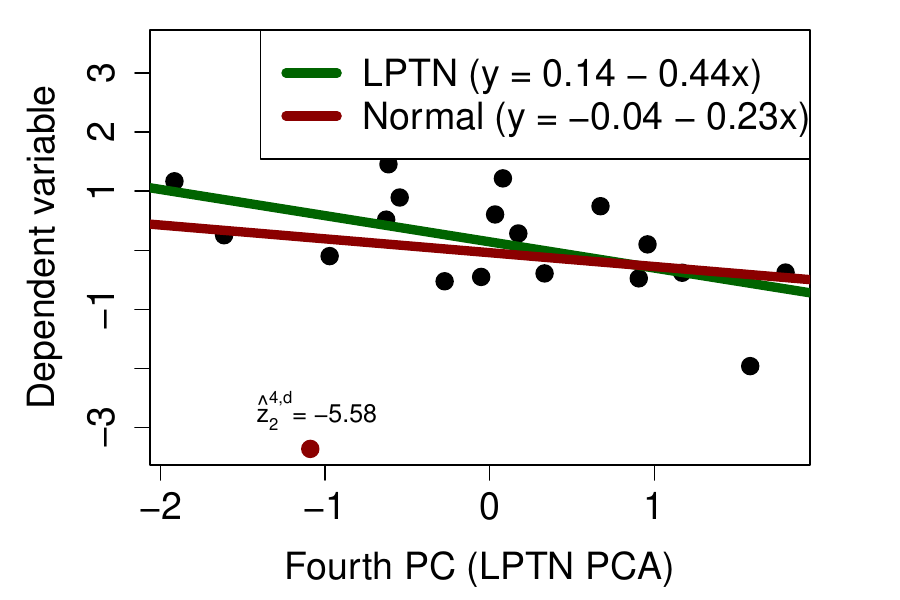} & \includegraphics[width=0.36\textwidth]{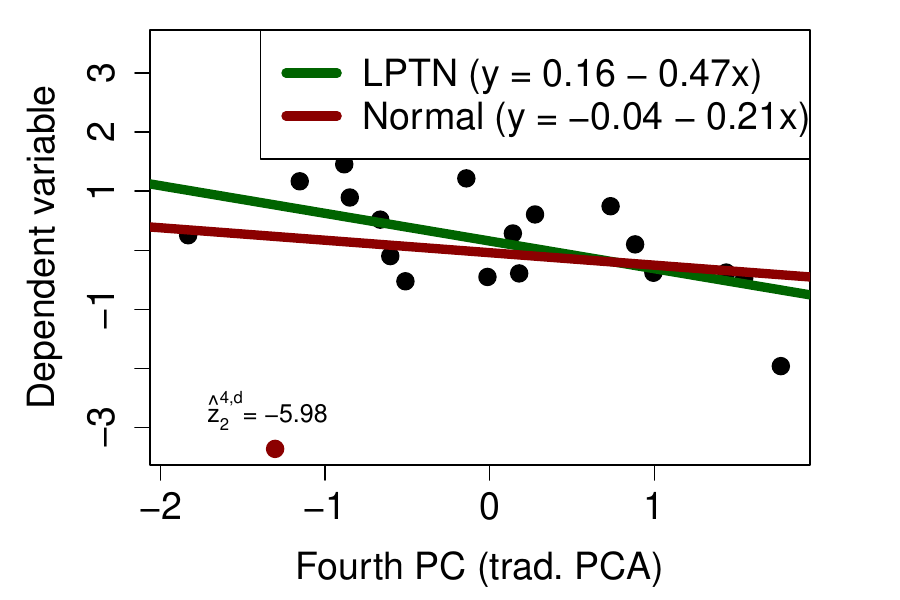} \cr
    \textbf{(c)} & \textbf{(d)} \cr
  \end{array}$
  \vspace{-3mm}
\caption{Linear relationships between the dependent variable and second PC constructed using the (a) robust PCA and (b) traditional PCA; linear relationships between the dependent variable and fourth PC constructed using the (c) robust PCA and (d) traditional PCA}\label{fig_4}
 \end{figure}

These discrepancies are ultimately seen to have an impact on the estimated linear regressions (\autoref{fig_4}).  In particular, combining the traditional PCA (instead of the robust one) with LPTN regression models gives an average absolute deviation of $0.53$ (instead of the value $0.49$ appearing in \autoref{table:realdata}). When implementing the robust PCR, only the second and fourth PC end up being retained for the second stage of the statistical analysis. This means that the models corresponding to $I_1:=\{1\}$ (intercept only), $I_2:=\{1,2\}$, and $I_3:=\{1,2,4\}$ are the only models considered in that stage. Their posterior probabilities each are $0.00$, $0.18$, and $0.82$. Note that the fourth PC is not selected by the normal (nonrobust) approach.


The difference between Bayesian and frequentist robust approaches not only resides in the selection of statistically significant PCs (subsequently used in prediction), but also in the PCA decomposition employed. In fact, if the robust frequentist approach were to use the same regressors as the Bayesian one, and then retaining the usual approach for selecting PCs and estimating parameters, the results would be (essentially) the same as those of the robust Bayesian approach. The results would also be (essentially) the same if the exact same PCs as those in  the Bayesian approach were selected (the frequentist approach selects an additional one, as explained in the next paragraph). Our analysis thus shows that, in our example, the proposed PCA represents the information from the covariates in lower dimensional spaces in a way that is more suitable to predict the dependent variable.

Our analysis also shows that, in the current example, the Bayesian approach for selecting statistically significant PCs does not dominate the frequentist one, and vice versa. The frequentist approach however leads to an overparameterised model, an undesirable characteristic. Indeed, frequentist approaches use cross-validation, along with a robust prediction measure in the case of the robust method. Models with an increasing number of PCs are thus evaluated (the PCs are ordered, and included in the models according to this predetermined order); the model enjoying the best fit is then used for prediction. If we apply this method using the PCs obtained from our robust PCA for instance, we find the model with four PCs to be the best option, and so according to this analysis, the first four PCs should be used for prediction. 
If we consider a larger class of models instead of being forced to include PCs in a predetermined manner, we however find through cross-validation methods that the model with the first, second and fourth PCs performs better (as pointed out by the proposed Bayesian method). Indeed, the third PC does not significantly explain the variability in the dependent variable; it may thus negatively influence the model's ability to generalise (see \cite{hadi1998cautionary} and \cite{jolliffe1982usePCR} for examples).

We note that the frequentist approach could be modified so as to include PCs on the merit of their individual contribution. To this effect, the Bayesian information criterion (BIC) could be used to evaluate the individual contributions of the PCs, after which the frequentist method could be applied on the retained ordered components only; the resulting set of PCs would be the one used in prediction. Similarly, our robust Bayesian PCR could also be applied under the frequentist paradigm.

We also ran simulations and drew the same conclusions as in this section. We thus do not present them for brevity. We only note that the robust approaches (Bayesian and frequentist) are expected to be efficient by their nature, and this is what we observed; they perform only slightly worse than their nonrobust counterparts when there are no outliers.

\section{Conclusion and further remarks}\label{sec_conclusion_pcr}

In light of the results of \autoref{sec_real_data}, we conclude that the proposed robust Bayesian PCR approach is expected to perform better than its competitors (at least those that are nonrobust) when there are outliers in the data set (either among the covariates or the dependent variable). This is a consequence of the new class of super heavy-tailed PCA models, combined to the LPTN regressions of \cite{gagnon2018regression}. The approach is also expected to perform better when the first $q$ PCs do not all contribute in explaining the variability of the dependent variable. The approach indeed takes advantage of the linear independence of the PCs to effectively exclude the components that are not relevant, and next forms a sequence of nested models from which predictions are produced and averaged out according to the model posterior probabilities.


As explained in \autoref{sec_pca_robust}, the robust PCA applied to the real data of \autoref{sec_real_data} is, in reality, an approximation to the exact wholly robust PCA model. Further research is needed to acquire a deeper understanding of its theoretical properties, as well as to develop an efficient implementation method. It would be particularly useful to obtain a robust procedure that not only reduces dimensionality, but also induces sparsity to deal with cases where $p\gg n$.

\section*{Acknowledgements}

The authors acknowledge support from NSERC (Natural Sciences and Engineering Research Council of Canada), FRQNT (Le Fonds de recherche du Québec - Nature et technologies) and SOA (Society of Actuaries). The authors thank two anonymous referees and an associate editor for suggestions that led to an improved paper.

\bibliographystyle{rss}
\bibliography{reference}

\section{Supplementary material}\label{sec_supp}

We first present the mathematical justification of the approximate robust principal component analysis (PCA) in \autoref{sec_just_PCA}. The validity of our prior structure is next addressed in \autoref{sec_validity_prior}. Propositions~\ref{prop_posterior} and \ref{prop_reversibility} are proved in \autoref{sec_proofs}. The list of the explanatory variables considered in the real data analysis in \autoref{sec_real_data} of our paper is provided in \autoref{sec_appendix}.

\subsection{Mathematical justification of the approximate robust PCA}\label{sec_just_PCA}

See Section 4.1 of our paper for the definition of notation. Given that the LPTN matches the normal distribution everywhere except in the tails, the limiting posterior of $(\tilde{\mathbf{Z}}_{q}, \tilde{\mathbf{L}}_q, \tilde{\mathbf{A}}_q,\eta)$ based on the exact robust PCA is similar to that arising from the traditional PCA model with normal errors based on $\mathbf{C}^*$, as the outliers moves away from the general trend. This means that the exact robust PCA applied to $\mathbf{C}$ leads to essentially the same singular value decomposition as a traditional PCA applied to $\mathbf{C}^*$ (in the limit). The approximate robust PCA method relies on this equivalence.

The first step in performing an approximate robust PCA is to obtain $\mathbf{C}$ by standardising the columns of the original data set. Location and scale estimates $\widehat{\mu}_j$ and $\widehat{\sigma}_j$ are thus used to standardise Column $j$, $j=1,\ldots,p$. Relying on a robust location-scale model as in \cite{desgagne2015robustness}, with an LPTN error distribution and $\rho:=0.95$, ensures that $(\widehat{\mu}_j, \widehat{\sigma}_j)\longrightarrow(\widehat{\mu}_{j}^{-\mathcal{O}}, \widehat{\sigma}_{j}^{-\mathcal{O}})$, where $(\widehat{\mu}_{j}^{-\mathcal{O}}, \widehat{\sigma}_{j}^{-\mathcal{O}})$ are estimates based on nonoutliers only. Note that the robust location-scale model is the linear regression model with only the intercept. For large $n$, we also have $(\widehat{\mu}_{j}^{-\mathcal{O}}, \widehat{\sigma}_{j}^{-\mathcal{O}})\approx (\widehat{\mu}_{j}^*, \widehat{\sigma}_{j}^*)$, where $(\widehat{\mu}_{j}^*, \widehat{\sigma}_{j}^*)$ are the sample mean and standard deviation obtained from $\mathbf{C}^*$, which is based on the normality of errors and in which outliers are replaced by their vertical projection. Denote by $c_{ij}^{\mathcal{O}}$ the outlying values; they are excluded for the estimation of $(\widehat{\mu}_{j}^{-\mathcal{O}}, \widehat{\sigma}_{j}^{-\mathcal{O}})$ and replaced by their vertical projection, denoted by $c_{ij}^*$, for the estimation of $(\widehat{\mu}_{j}^*, \widehat{\sigma}_{j}^*)$.
 Provided that $n$ is large enough, the impact of those points on the sample mean and standard deviation will indeed be negligible; furthermore, it was previously argued that estimates obtained under LPTN and normal error distributions are similar. The resulting matrices $\mathbf{C}$ and $\mathbf{C}^*$ are the same in the limit, except on lines containing outliers. 

The second step in performing the approximate robust PCA consists in computing robust correlations between all pairs of columns in $\mathbf{C}$. We know that the correlation between the standardised columns $j_1$ and $j_2$ of $\mathbf{C}^*$  is $\widehat{\beta}_{j_1,j_2}^\mathcal{N}$, the OLS slope estimate. We are interested in comparing the robust slope estimator (applied to columns of the matrix $\mathbf{C}$) to $\widehat{\beta}_{j_1,j_2}^\mathcal{N}$. When using a robust regression model as in \cite{gagnon2018regression} with an LPTN error distribution and $\rho:=0.95$, we find $\widehat{\beta}_{j_1,j_2}\longrightarrow \widehat{\beta}_{j_1,j_2}^{-\mathcal{O}}$, where $\widehat{\beta}_{j_1,j_2}^{-\mathcal{O}}$ is the robust slope estimate obtained using nonoutliers only. Again, for large $n$, we find $\widehat{\beta}_{j_1,j_2}^{-\mathcal{O}}\approx \widehat{\beta}_{j_1,j_2}^\mathcal{N}$. The robust correlation matrix obtained from $\mathbf{C}$ is thus asymptotically equal to the correlation matrix obtained from $\mathbf{C}^*$. Its diagonal elements are equal to 1; for simplicity, we set the upper diagonal entries to $\widehat{\beta}_{j_1,j_2}$, where Column $j_2$ plays the role of the dependent variable; we then make the matrix symmetrical.

The PCs $\widehat{\mathbf{Z}}_{q}$ are ultimately computed using $\mathbf{C} \widehat{\mathbf{v}}_j$, with $\widehat{\mathbf{v}}_j$ being the $j$-th eigenvector of the robust correlation matrix of $\mathbf{C}$.

\subsection{Validity of our prior structure}\label{sec_validity_prior}

Relying on improper priors such as $\pi(\sigma_k, \boldsymbol\beta_k\mid  k)=c_k/\sigma_k$ may lead to inconsistencies in model selection (see \cite{casella2009consistency}). For instance, one could select different constants $c_k$ in different models so as to yield the desired conclusions. 
In this section, we show that the Jeffreys-Lindley paradox does not arise in our PCR context under the normal distribution assumption.  It is thus expected to not arise either under the robust LPTN distribution, given its similarity to the normal.

Consider Models $s$ and $t$, where Model $s$ is nested in Model $t$. The ratio of the posterior probabilities of these two models is given by (see Proposition~3.1 in our paper)
\begin{align}\label{eqn_bayes_factor}
\frac{\pi(t\mid \mathbf{y})}{\pi(s\mid\mathbf{y})} &= \frac{\Gamma((n-d_s)/2 - (d_t - d_s)/2)}{\Gamma((n-d_s)/2) ((n - d_s) / 2)^{-(d_t - d_s) / 2}} \, n^{-(d_t - d_s) / 2} \, \left(\frac{\|\mathbf{y} - \widehat{\mathbf{y}}_s\|^2/(n - 1)}{\|\mathbf{y} - \widehat{\mathbf{y}}_t\|^2/(n - 1)}\right)^{n/2}  \cr
& \qquad \times \frac{\pi^{d_t/2}}{\pi^{d_s/2}} \, \frac{((n - d_s) / 2)^{-(d_t - d_s) / 2}}{n^{-(d_t - d_s) / 2}} \, \frac{\left(\|\mathbf{y} - \widehat{\mathbf{y}}_t\|^2/(n - 1)\right)^{d_t/2}}{\left(\|\mathbf{y} - \widehat{\mathbf{y}}_s\|^2/(n - 1)\right)^{d_s/2}} \, \frac{\pi(t)}{\pi(s)}.
\end{align}
The difference between the Bayesian information criteria (BIC, \cite{schwarz1978estimating}) of Models $t$ and $s$ is given by
\begin{align*}
 \text{BIC}_t - \text{BIC}_s &= n\log\left(\|\mathbf{y} - \widehat{\mathbf{y}}_t\|^2/n\right) + (d_t + 1)\log n \cr
 &\qquad - n\log\left(\|\mathbf{y} - \widehat{\mathbf{y}}_s\|^2/n\right) - (d_s + 1)\log n  \cr
 &= n \log\left(\frac{\|\mathbf{y} - \widehat{\mathbf{y}}_t\|^2/n}{\|\mathbf{y} - \widehat{\mathbf{y}}_s\|^2/n}\right) + (d_t - d_s)\log n.
\end{align*}
Given that the first ratio on the right hand side of \eqref{eqn_bayes_factor} converges to 1 as $n\longrightarrow\infty$, we have that $\exp\{-(\text{BIC}_t - \text{BIC}_s)/2\}$ asymptotically behaves like the first part on the right hand side of \eqref{eqn_bayes_factor}. The terms $\left(\|\mathbf{y} - \widehat{\mathbf{y}}_k\|^2/(n-1)\right)^{d_k/2}$ in \eqref{eqn_bayes_factor} converge towards a constant (in $n$) and are thus dominated. The other terms in \eqref{eqn_bayes_factor} are either constant in terms of $n$ or dominated as well. Therefore, $\pi(t| \mathbf{y})/\pi(s|\mathbf{y})$ and $\exp\{-(\text{BIC}_t - \text{BIC}_s)/2\}$ share the same asymptotic behaviour. This will be sufficient to prove that the prior structure does not prevent the Bayesian variable selection procedure to be consistent, in the same sense as \cite{casella2009consistency}. If the ``true'' model is among the models considered, then its posterior probability converges to 1 as $n$ increases. Further technical details are required for a rigorous proof. Empirical evidences also point towards the validity of our claim. 

It would be interesting to investigate the asymptotic behaviour in the more general context of traditional linear regression. The fact that the regressors are standardised and linearly independent plays a role in the sketch of the proof presented above. It would however be surprising if a similar prior structure, but with slightly correlated standardised regressors, led to inconsistencies.

In practice (with finite samples), one may set the prior $\pi(k)$ to be proportional to $\pi^{-d_k/2}$ times a prior opinion about $\left(\|\mathbf{y} - \widehat{\mathbf{y}}_k\|^2/(n - 1)\right)^{-d_k/2}$, to cancel the effect of these two terms in \eqref{eqn_bayes_factor}. In the numerical analyses, we set $\pi(k)\propto 1$ because we do not have relevant information. Note that the robust approach proposed in this paper can be used with any informative prior such as those in \cite{raftery1997bayesian}.

\subsection{Proofs}\label{sec_proofs}

 \begin{proof}[Proof of Proposition~3.1]
  The proof is essentially a computation using that $f:=\mathcal{N}(0,1)$ and the structure of the principal components. First,
  \begin{align*}
  \pi(k,\sigma_k,\boldsymbol\beta_k\mid \mathbf{y})&\propto f(\mathbf{y}\mid k,\sigma_k,\boldsymbol\beta_k)\pi(\sigma_k,\boldsymbol\beta_k\mid k)\pi(k) \cr
  &\propto f(\mathbf{y}\mid k,\sigma_k,\boldsymbol\beta_k)(1/\sigma_k)\pi(k).
 \end{align*}
 The likelihood function for a given model is
\begin{align*}
 f(\mathbf{y}\mid k,\sigma_k,\boldsymbol\beta_k)&=\prod_{i=1}^n \frac{1}{\sigma_k \sqrt{2\pi}}\exp\left\{-\frac{1}{2\sigma_k^2}(y_i-\mathbf{x}_{i,k}^T \boldsymbol\beta_k)^2\right\} \cr
 &=\frac{1}{\sigma_k^n (2\pi)^{n/2}} \exp\left\{-\frac{1}{2\sigma_k^2}\sum_{i=1}^n(y_i-\mathbf{x}_{i,k}^T \boldsymbol\beta_k)^2\right\}.
\end{align*}
  We now analyse the sum in the exponential:
 \begin{align*}
  \sum_{i=1}^n(y_i-\mathbf{x}_{i,k}^T \boldsymbol\beta_k)^2&= \sum_{i=1}^n y_i^2-2\sum_{i=1}^n y_i\sum_{j=1}^{d_k} x_{iI_{j,k}} \beta_{j,k}+\sum_{i=1}^n \left(\sum_{j=1}^{d_k} x_{iI_{j,k}}\beta_{j,k}\right)^2 \cr
  &=n-1-2\sum_{j=1}^{d_k} \beta_{j,k} \sum_{i=1}^n y_i x_{iI_{j,k}}+\sum_{i=1}^n \left(\sum_{j=1}^{d_k} x_{i I_{j,k}}\beta_{j,k}\right)^2,
 \end{align*}
using that $\sum_{i=1}^n y_i^2 = n-1$. We also have
 \begin{align*}
  \sum_{i=1}^n \left(\sum_{j=1}^{d_k} x_{iI_{j,k}}\beta_{j,k}\right)^2&=\sum_{i=1}^n \left(\sum_{j=1}^{d_k} (x_{i I_{j,k}}\beta_{j,k})^2+\sum_{j,s=1 (j\neq s)}^{d_k} x_{iI_{j,k}}\beta_{j,k} x_{iI_{s,k}}\beta_{s,k}\right) \cr
  &= \sum_{j=1}^{d_k} \beta_{j,k}^2 \sum_{i=1}^n x_{i I_{j,k}}^2,
 \end{align*}
 using $\sum_{i=1}^n x_{ij}x_{is}=0$ for all $j,s\in\{2,\ldots,d\}$ with $j\neq s$, $x_{11}=\ldots=x_{n1}=1$, $(1/n)\sum_{i=1}^n x_{ij}=0$ for all $j\in\{2,\ldots,d\}$. Consequently,
 \begin{align*}
  \sum_{i=1}^n(y_i-\mathbf{x}_{i,k}^T \boldsymbol\beta_k)^2&= n-1-2\sum_{j=1}^{d_k} \beta_{j,k} \sum_{i=1}^n y_i x_{iI_{j,k}}+ \sum_{j=1}^{d_k} \beta_{j,k}^2 \sum_{i=1}^n x_{i I_{j,k}}^2  \cr
  &=n-1- \ind(k\geq 2)2\sum_{j=2}^{d_k} \beta_{j,k} \sum_{i=1}^n y_i x_{iI_{j,k}}+n\beta_{1,k}^2 \cr
  &\qquad+\ind(k\geq 2)(n-1)\sum_{j=2}^{d^k} \beta_{j,k}^2 ,
 \end{align*}
 using again $x_{11}=\ldots=x_{n1}=1$, $\sum_{i=1}^n y_i =0$ and $\sum_{i=1}^n x_{ij}^2=n - 1$ for all $j\in\{2,\ldots,d\}$. We also have
 \begin{align*}
  &\ind(k\geq 2)\left((n-1)\sum_{j=2}^{d_k} \beta_{j,k}^2 -2\sum_{j=2}^{d_k} \beta_{j,k} \sum_{i=1}^n y_i x_{iI_{j,k}}\right) \cr
  &­\quad=\ind(k\geq 2)(n-1)\sum_{j=2}^{d_k} \left(\beta_{j,k}^2-2\beta_{j,k} \frac{\sum_{i=1}^n y_i x_{iI_{j,k}}}{n-1}\right) \cr
  &\quad=\ind(k\geq 2)(n-1)\sum_{j=2}^{d_k} \left(\beta_{j,k}-\frac{\sum_{i=1}^n x_{iI_{j,k}} y_i}{n-1}\right)^2 \cr
  &\qquad-\ind(k\geq 2) (n-1) \sum_{j=2}^{d_k} \left(\frac{\sum_{i=1}^n x_{i I_{j,k}} y_i}{n-1}\right)^2.
 \end{align*}
 Putting this together leads to:
 \begin{align*}
      &\pi(k,\sigma_k,\boldsymbol\beta_k\mid \mathbf{y}) \cr
      &\propto \pi(k)(2\pi)^{d_k/2} \, \frac{1}{\sigma_k^{n-{d_k}+1}} \exp\left\{-\frac{n-1}{2 \sigma_k^2}\left(1 -\ind(k\geq 2)\sum_{j\in I_k\setminus\{1\}} \left(\frac{\sum_{i=1}^n x_{ij} y_i}{n-1}\right)^2\right)\right\} \cr
      &\quad \times \frac{1}{\sigma_k \sqrt{2\pi}} \exp\left\{-\frac{n}{2\sigma_k^2}\beta_{1,k}^2\right\} \cr
      &\quad \times\left(\ind(k=1)+\ind(k\geq 2) \prod_{j=2}^{d_k}  \frac{1}{\sigma_k \sqrt{2\pi}} \exp\left\{-\frac{n-1}{2 \sigma_k^2}\left(\beta_{j,k}-\frac{\sum_{i=1}^n x_{iI_{j,k}} y_i}{n-1}\right)^2\right\}\right).
  \end{align*}
  We multiply and divide by the appropriate terms. The only remaining thing to show is that
  \[
   n-1 \left(1 -\ind(k\geq 2)\sum_{j\in I_k\setminus\{1\}} \left(\frac{\sum_{i=1}^n x_{ij} y_i}{n-1}\right)^2\right)=\|\mathbf{y} - \widehat{\mathbf{y}}_k\|_2^2.
  \]
  Firstly, $n-1=\|\mathbf{y} \|_2^2$. Also,
  \begin{align*}
   \|\mathbf{y} \|_2^2 &= \|\mathbf{y} - \widehat{\mathbf{y}}_k + \widehat{\mathbf{y}}_k\|_2^2 \cr
   &=  \|\mathbf{y} - \widehat{\mathbf{y}}_k\|_2^2 + (\mathbf{y} - \widehat{\mathbf{y}}_k)^T \widehat{\mathbf{y}}_k + \widehat{\mathbf{y}}_k^T (\mathbf{y} - \widehat{\mathbf{y}}_k) + \widehat{\mathbf{y}}_k^T \widehat{\mathbf{y}}_k.
  \end{align*}
  We know that $(\mathbf{y} - \widehat{\mathbf{y}}_k)^T \widehat{\mathbf{y}}_k=\widehat{\mathbf{y}}_k^T (\mathbf{y} - \widehat{\mathbf{y}}_k)=0$ because $\mathbf{y} - \widehat{\mathbf{y}}_k$ is the vector of residuals which is orthogonal to $\widehat{\mathbf{y}}_k$. Finally,
  \begin{align*}
   \widehat{\mathbf{y}}_k^T \widehat{\mathbf{y}}_k = (\mathbf{X}_k\widehat{\boldsymbol\beta}_k)^T \mathbf{X}_k\widehat{\boldsymbol\beta}_k = \widehat{\boldsymbol\beta}_k^T \mathbf{X}_k^T \mathbf{X}_k \widehat{\boldsymbol\beta}_k &= (n-1) \|\widehat{\boldsymbol\beta}_k\|_2^2 \cr
   &= (n-1) \ind(k\geq 2)\sum_{j\in I_k\setminus\{1\}} \left(\frac{\sum_{i=1}^n x_{ij} y_i}{n-1}\right)^2,
  \end{align*}
  where $\mathbf{X}_k$ is the design matrix associated with Model $k$.
 \end{proof}

\begin{proof}[Proof of Proposition~2.2]
As explained in \cite{green1995reversible}, it suffices to separately verify that the probability to go from a set $A$ to a set $B$ is equal to the probability to go from $B$ to $A$ when updating the parameters and when switching models, for accepted movements and for any appropriate $A,B$.

When updating the parameters, the probability to go from a set $A$ to a set $B$ is given by
\begin{align*}
 &\int_A \pi(k,\sigma_k,\boldsymbol\beta_k\mid \mathbf{y})  g(1)\int_B  \prod_{i=1}^{1+d_k} \varphi_i(w_i\mid k,(\sigma_k,\boldsymbol\beta_k)_i,\ell_k) \cr
 &\hspace{50mm}\times\left(1\wedge\frac{(1/w_1) f(\mathbf{y}\mid k,\mathbf{w}_{k})}{(1/\sigma_{k})f(\mathbf{y}\mid k,\sigma_{k},\boldsymbol\beta_{k})}\right)d\mathbf{w}_k \, d(\sigma_k,\boldsymbol\beta_k).
\end{align*}
  Using Fubini's theorem, this probability is equal to
\begin{align*}
 &\int_B \pi(k,\mathbf{w}_k\mid \mathbf{y}) g(1)\int_A  \prod_{i=1}^{1+d_k} \varphi_i((\sigma_k,\boldsymbol\beta_k)_i\mid k,w_i,\ell_k) \cr
 &\hspace{50mm}\times\left(1\wedge\frac{(1/\sigma_{k})f(\mathbf{y}\mid k,\sigma_{k},\boldsymbol\beta_{k})}{(1/w_1) f(\mathbf{y}\mid k,\mathbf{w}_{k})}\right) d(\sigma_k,\boldsymbol\beta_k)\,d\mathbf{w}_k,
\end{align*}
which is the probability to go from $B$ to $A$. Note that this is valid for all $k\in\{1,\ldots, $ $\text{K}_{\text{max}}\}$.

The probability to switch from Model $k\in\{1,\ldots,\text{K}_{\text{max}}-1\}$, where the parameters are in the set $A$, to Model $k+1$, where the parameters are in the set $A'\times B$ (the set $A'$ is a modified version of $A$ to account for the addition of $\mathbf{c}_{k+1}$), is given by
\begin{align*}
 &\int_A \pi(k,\sigma_k,\boldsymbol\beta_k\mid \mathbf{y})   g(2)\int_B q_{k+1}(u_{k+1}) \cr
 &\qquad\times\left(1\wedge\frac{\pi(k+1)f(\mathbf{y}\mid k+1,(\sigma_{k},\boldsymbol\beta_{k})+\mathbf{c}_{k+1},u_{k+1})}{\pi(k)f(\mathbf{y}\mid k,\sigma_{k},\boldsymbol\beta_{k})q_{k+1}(u_{k+1})}\right)
 du_{k+1}\, d(\sigma_k,\boldsymbol\beta_k).
\end{align*}
After the change of variables $(\sigma_{k+1},\boldsymbol\beta_{k+1})=((\sigma_{k},\boldsymbol\beta_{k})+\mathbf{c}_{k+1},u_{k+1})$, we have
\begin{align*}
 &\int_{A'\times B} \pi(k,(\sigma_{k+1},\boldsymbol\beta_{k+1}^-)-\mathbf{c}_{k+1}\mid \mathbf{y})   g(2) q_{k+1}(\beta_{d_{k+1},k+1}) \cr
 &\qquad\times\left(1\wedge\frac{\pi(k+1)f(\mathbf{y}\mid k+1,\sigma_{k+1},\boldsymbol\beta_{k+1})}{\pi(k)f(\mathbf{y}\mid k,(\sigma_{k+1},\boldsymbol\beta_{k+1}^-)-\mathbf{c}_{k+1})q_{k+1}(\beta_{d_{k+1},k+1})}\right)
 d(\sigma_{k+1},\boldsymbol\beta_{k+1}).
\end{align*}
This last probability is equal to
 \begin{align*}
 &\int_{A'\times B} \pi(k+1,\sigma_{k+1},\boldsymbol\beta_{k+1}\mid \mathbf{y})   g(3) \cr
 &\qquad\times\left(1\wedge\frac{\pi(k)f(\mathbf{y}\mid k,(\sigma_{k+1},\boldsymbol\beta_{k+1}^-)-\mathbf{c}_{k+1})q_{k+1}(\beta_{d_{k+1},k+1})}{\pi(k+1)f(\mathbf{y}\mid k+1,\sigma_{k+1},\boldsymbol\beta_{k+1})}\right)
 d(\sigma_{k+1},\boldsymbol\beta_{k+1}),
\end{align*}
which is the probability to switch from Model $k+1$, where the parameters are in the set $A'\times B$, to Model $k$, where the parameters are in the set $A$.

Therefore, the Markov chain $\{(K,\sigma_{K},\boldsymbol\beta_{K})(m): m\in\na\}$ satisfies the reversibility condition with respect to the posterior.
\end{proof}

\subsection{List of the explanatory variables used in Section 6}\label{sec_appendix}

  \begin{table}[H]
  \centering
   \begin{tabular}{ll}
    \toprule
    \textbf{Name} & \textbf{Ticker symbol} \cr
    \midrule
      Artis Real Estate Investment Trust & AX-UN.TO  \cr
      Asanko Gold Inc. & AKG.TO \cr
      Bonterra Energy Corp. & BNE.TO \cr
      Canadian Imperial Bank Of Commerce & CM.TO \cr
      CI Financial Corp. & CIX.TO \cr
      Celestica Inc. Subordinate Voting Shares & CLS.TO \cr
      DHX Media Ltd. & DHX-B.TO \cr
      Dominion Diamond Corporation & DDC.TO \cr
      Gildan Activewear Inc. & GIL.TO \cr
      Husky Energy Inc. & HSE.TO \cr
      iPath Bloomberg Sugar Subindex & SGG \cr
      iShares MSCI Japan &  EWJ \cr
      iShares 20+ Year Treasury Bond & TLT \cr
      Laurentian Bank of Canada & LB.TO \cr
      Parkland Fuel Corporation & PKI.TO \cr
      United States Oil Fund LP & USO \cr
      Vermilion Energy Inc. & VET.TO \cr
      Volume of the S\&P 500 & N/A \cr
    \bottomrule
 \end{tabular}
 \caption{Names of the companies, funds, and financial indicators used as explanatory variables in the analysis in Section~6 of our paper, with their ticker symbol (if available)} \label{tab_names}
 \end{table}

\end{document}